\documentclass{article}
\usepackage[a4paper,margin=2.9cm]{geometry}
\usepackage[affil-it]{authblk}
\usepackage{amsmath,amsfonts,amssymb}
\usepackage{algorithmic}
\usepackage{algorithm}
\usepackage{array}
\usepackage[caption=false,font=normalsize,labelfont=sf,textfont=sf]{subfig}
\usepackage{graphicx}
\usepackage{calc}
\usepackage[normalem]{ulem} 

\usepackage[utf8]{inputenc}
\usepackage{authblk}

\usepackage[
  bibstyle=phys,
  citestyle=numeric-comp,
  sorting=none,
  backend=biber
]{biblatex}

\DeclareFieldFormat{labelnumber}{\mkbibbrackets{#1}}

\defbibenvironment{bibliography}
  {\list
     {\printfield{labelnumber}\hspace{0.5em}} 
     {
      \setlength{\labelwidth}{\widthof{[999]}}
      \setlength{\leftmargin}{\labelwidth}%
      \setlength{\labelsep}{\biblabelsep}%
      \addtolength{\leftmargin}{\labelsep}%
      \setlength{\itemsep}{\bibitemsep}%
      \setlength{\parsep}{\bibparsep}%
     }
  }
  {\endlist}
  {\item}

\usepackage{amsthm}
\usepackage{changepage}
\usepackage{enumerate}
\usepackage{color}
\usepackage{physics}
\usepackage{parcolumns}
\usepackage{mathrsfs}
\usepackage{multirow}
\usepackage{tikz-cd}
\usepackage{environ}
\usepackage{tabularx}
\usepackage{complexity}
\usepackage[T1]{fontenc}
\usepackage{comment}
\usepackage{todonotes}
\usepackage[unicode=true,
    pdfpagelabels,
    hypertexnames=true,
    colorlinks,
    linkcolor=black,
    citecolor=black,
    urlcolor=black]{hyperref}
    \hypersetup{
     colorlinks = true,
     citecolor = magenta,
     linkcolor = magenta,
     urlcolor = magenta}  
\usepackage[capitalise]{cleveref}

\theoremstyle{plain}
\newtheorem{thm}{Theorem}
\newtheorem{prop}{Proposition}
\newtheorem{lemma}{Lemma}
\newtheorem{coro}{Corollary}

\newtheorem{example}{Example}

\theoremstyle{definition}
\newtheorem{remark}{Remark}

\newcommand{\id}{\operatorname{id}}
\newcommand{\pe}{\tilde{+}}
\newcommand{\me}{\tilde{-}}
\newcommand{\tht}{\frac{\theta}{2}}

\definecolor{amethyst}{rgb}{0.75, 0, 1}

\definecolor{cinnamon}{rgb}{0.7, 0.3, 0.2}

\newcommand{\mc}{\mathcal}
\renewcommand{\H}{\mathcal{H}}

\newcommand{\B}{\mathcal{B}}

\renewcommand{\S}{\mc{S}}

\newcommand{\fmig}{\frac{1}{2}}
\newcommand{\I}{\mathbb{I}}
\newcommand{\ii}{\mathrm{i}}

\newcommand{\swich}[1]{\left(#1\right)}
\newcommand{\cwich}[1]{\left[#1\right]}
\newcommand{\set}[1]{\left\{#1\right\}}

\newcommand{\kket}[1]{\ket{\mkern-2mu\ket{#1}\mkern-4mu}}

\newcommand{\kketbbra}[1]{\ket{\mkern-2mu\ket{#1}\mkern-4mu}\mkern-5mu\bra{\mkern-4mu\bra{#1}\mkern-2mu}}

\newcommand{\supp}{\operatorname{supp}}

\bibliography{biblio2.bib}

\makeatletter
\let\inserttitle\@title
\makeatother

\title{\textbf{Single-letter Chain Rule for Quantum Relative Entropy}}
\date{\vspace{-10ex}}

\author[1,2]{Giulio Gasbarri\thanks{\href{mailto:giulio.gasbarri@uni-siegen.de}{giulio.gasbarri@uni-siegen.de}}$^,$\textsuperscript{$\ddagger$}$^,$}
\author[2,3]{Matt Hoogsteder-Riera\thanks{\href{mailto:matt.hoogsteder@uab.cat}{matt.hoogsteder@uab.cat}}$^,$\textsuperscript{$\ddagger$}$^,$}

\affil[1]{Naturwissenschaftlich-Technische Fakult\"{a}t, \protect\\Universit\"{a}t Siegen, 57068 Siegen, Germany}
\affil[2]{{\textit{Grup d'Informaci\'o Qu\`antica, Departament de F\'isica,\protect\\[-1mm] Universitat Aut\`onoma de Barcelona, 08193 Bellaterra (BCN), Spain}}}
\affil[3]{Department Mathematik/Informatik--Abteilung Informatik,\protect\\[-1mm] Universit\"{a}t zu K\"{o}ln, Albertus-Magnus-Platz, 50923 K\"{o}ln, Germany}

\begin{document}

\maketitle
\def\thefootnote{$\ddagger$}\footnotetext{The authors contributed equally to this work.}
\def\thefootnote{\arabic{footnote}}

    \vspace{2em} 
\begin{abstract}
    
Relative entropy is the standard measure of distinguishability in classical and quantum information theory. In the classical case, its loss under channels admits an exact chain rule, while in the quantum case only asymptotic, regularized chain rules are known. We establish new chain rules for quantum relative entropy that apply already in the single-copy regime. The first inequality can be naturally obtained via POVM decompositions, extending the point distributions in the classical chain rule to quantum ensemble partitions. The second gives a sufficient condition for the most natural extension of the classical result, which uses projectors as an analog for the classical point distributions.
We additionally find a semiclassical chain rule where the point distributions are replaced with the projectors of the initial states, and, finally, we find a relation to previous works on strengthened data processing inequalities and recoverability.
These results show that meaningful chain inequalities are possible already at the single-copy level, but they also highlight that tighter bounds remain to be found.
\end{abstract}

\tableofcontents

\section{Introduction}

Entropy inequalities formalize the principle that statistical distinguishability cannot increase under noisy processing. In the classical setting, distinguishability is quantified by the Kullback-Leibler (KL) divergence between probability distributions $p,q$ on a finite alphabet $\mc{X}$,
\begin{align}
D(p\|q) = \sum_x p(x)\,\log \frac{p(x)}{q(x)}.
\end{align}
The KL divergence is non-negative, vanishes if and only if $p=q$, and obeys the \emph{data-processing inequality} (DPI), i.e.
\begin{align}
D(Mp\|Mq) \leq D(p\|q),
\end{align}
for every stochastic map $M$. 
Beyond monotonicity, the classical case admits a sharper refinement in the form of a chain rule. For any pair of stochastic maps $M$ and $N$,
\begin{align}\label{eq:cl_Chain_rule}
D(p \|q ) - D(Mp\|Nq) \;\geq\; - \mathbb{E}_p\, D(M\delta_j \| N\delta_j),
\end{align}
where $\delta_j$ denotes the point distribution at $j$.
This inequality decomposes the global change of relative entropy 
into an average of local divergences of the point distributions. 
It provides a finer description of how distinguishability flows through channels and underpins structural properties such as joint convexity of relative entropy.

In the quantum setting, probability distributions are replaced by density operators. For states $\rho,\;\sigma$ on a finite-dimensional Hilbert space, the Umegaki relative entropy~\cite{hiai91,ogawa01} is defined as
\begin{align}
D(\rho\|\sigma) =
\begin{cases}
\Tr\!\left[\rho (\log \rho - \log \sigma)\right], & \text{if }\supp(\rho)\subseteq\supp(\sigma), \\[4pt]
+\infty, & \text{otherwise}.
\end{cases}
\end{align}
Its monotonicity under quantum channels,
\begin{align}
D(\rho\|\sigma) \geq D(\mc{M}(\rho)\|\mc{M}(\sigma)),
\end{align}
is the \emph{quantum DPI}, established independently by Lindblad~\cite{lindblad75} and Uhlmann~\cite{uhlmann77}, and known to be equivalent to the strong subadditivity of entropy~\cite{lieb73}.

The equality conditions for the DPI were characterized by Petz~\cite{petz86,petz88}, who introduced the recovery map
\begin{align}
\mc{R}_{\sigma,\mc{M}}(X) = \sigma^{1/2}\,\mc{M}^\dagger\!\left(\mc{M}(\sigma)^{-1/2}\,X\,\mc{M}(\sigma)^{-1/2}\right)\,\sigma^{1/2},
\end{align}
and proved that equality holds if and only if $\rho = \mc{R}_{\sigma,\mc{M}}(\mc{M}(\rho))$. This result initiated the systematic study of recoverability.
Fawzi and Renner~\cite{fawzi15} gave the first quantitative refinement, proving that the conditional mutual information of a tripartite state lower bounds the fidelity of recovery:
\begin{align}
I(A:C|B)_\rho \;\geq\; -2\log F\!\left(\rho_{ABC}, \mc{R}_{B\to BC}(\rho_{AB})\right).
\end{align}
This revealed that small conditional mutual information guarantees the existence of a high-fidelity recovery channel. 

An important refinement came with the recognition that the~\emph{measured relative entropy},
\begin{align}
\mathbb{D}_M(\rho\|\sigma)=\sup_M D_M(\rho\|\sigma) =  \sup_M D(P_\rho^{M}\|P_\sigma^{M}),
\end{align} is the correct quantity in strengthened inequalities. Here $P_\rho^{M}$ is the probability distribution generated by measuring POVM $M=\{M_i\}$ on $\rho$. Brand\~{a}o, Harrow, Aram, Oppenheim and Strelchuck~\cite{brandao2015quantum} established
\begin{align}
I(A:C|B)_\rho \;\geq\; \mathbb{D}_M\!\left(\rho_{ABC}\,\big\|\, \mc{R}_{B\to BC}(\rho_{AB})\right).
\end{align}
This result demonstrates that the conditional mutual information upper-bounds the optimal measured relative entropy between the true tripartite state and any state recovered from its marginal on~$AB$, optimized over all possible measurements.

These ideas were subsequently extended to quantum relative entropy.
Wilde~\cite{wilde15} and Sutter, Tomamichel and Harrow~\cite{sutter16strengthened} established strengthened data-processing inequalities that bound the loss of relative entropy under a quantum channel.
Their results make use of rotated Petz recovery maps, defined for each real parameter $t \in \mathbb{R} $ by
\begin{align}
\mc{R}_{\sigma,\mc{M}}^{(t)}(X) = 
\sigma^{\frac{1+it}{2}}\,
\mc{M}^\dagger\!\left(\mc{M}(\sigma)^{-\frac{1+it}{2}}\,X\,\mc{M}(\sigma)^{-\frac{1-it}{2}}\right)\,
\sigma^{\frac{1-it}{2}}.
\end{align} 
 Building on Wilde's work~\cite{wilde15}, Sutter~\emph{et al.}~\cite{sutter16strengthened} employed the rotated recovery map to prove that 
\begin{align}\begin{split}
\label{eq:recoveryBounds}
D(\rho\|\sigma) - D(\mc{M}(\rho)|\mc{M}(\sigma))&\geq \mathbb{D}_M(\rho\|\mc{R}\circ \mc{M}(\rho)) \\
&\geq -\,2 \log F\!\left(\rho,\; \mc{R}\circ\mc{M}(\rho)\right).
\end{split}\end{align} 
In these strenghtened DPI the recovery map $\mc{R}$ is a $\rho$-dependent convex combination of rotated Petz maps. This naturally raises the question of whether a universal (i.e., $\rho$-independent) recovery map exists.

Junge, Renner, Sutter, Wilde, and Winter~\cite{junge18} and Sutter, Berta, and Tomamichel \cite{sutter16} answered this positively, by introducing an explicit convex mixture 
\begin{align}
    \overline{\mc{R}}_{\sigma,\mc{M}}=\int_{-\infty}^{\infty}d\beta_0(t)\;\mc{R}_{\sigma,\mc{M}}^{(t)},\qquad \beta_0(t)=\frac{\pi}{2(1+\cosh(\pi t))}
\end{align}  
of rotated Petz maps,
\begin{align}
\mc{R}_{\sigma,\mc{M}}^{(t)}(X) =
\sigma^{\frac{1+it}{2}}\,
\mc{M}^\dagger\!\left(\mc{M}(\sigma)^{-\frac{1+it}{2}}\,X\,\mc{M}(\sigma)^{-\frac{1-it}{2}}\right)\,
\sigma^{\frac{1-it}{2}}.
\end{align}
Junge \emph{et al.}~\cite{junge18} improved the fidelity bound by showing that the loss of quantum relative entropy under any quantum channel admits a quantitative lower bound in terms of this universal recovery map:
\begin{align}\label{eq:universalBound}
D(\rho\|\sigma) - D(\mc{M}(\rho)\|\mc{M}(\sigma)) \;\geq\; -2\log F\!\!\left(\rho\,\big\|\,\overline{\mc{R}}_{\sigma,\mc{M}}\circ\mc{M}(\rho)\right).
\end{align}
thereby establishing a sharp relation between entropy loss and recoverability valid for arbitrary channels.

Subsequently, Sutter, Berta, and Tomamichel~\cite{sutter16} proved a universal strengthened DPI in terms of the measured relative entropy:
\begin{align}\label{eq:strineq}
D(\rho\|\sigma) - D(\mc{M}(\rho)\|\mc{M}(\sigma))\ge \mathbb{D}_M(\rho\|\overline{\mc{R}}_{\sigma,\mc{M}}\circ \mc{M}(\rho)).
\end{align} 
This result quantitatively relates entropy loss to recoverability for arbitrary channels.
Unlike in the classical setting, where exact refinements hold in terms of relative entropy, explicit counterexamples show that quantum strengthened inequalities in terms of relative entropy $D$ are not possible; only the measured relative entropy $\mathbb{D}_{M}$ yields valid bounds~\cite{fawzi18eff,hirche18Thesis}.
While strengthened recoverability inequalities are now well understood, no quantum analogue of the classical chain rule~\cref{eq:cl_Chain_rule} is known.
In the quantum setting, noncommutativity obstructs a decomposition of relative entropy into local divergences of point distributions. 
 Fang, Fawzi, Renner, and Sutter~\cite{fang20} made this precise observing that the naive inequality with the single-letter channel divergence $D(\mc{M}\|\mc{N})$ does not hold. They furthermore showed that a meaningful chain rule can be recovered in the many-copy setting by introducing the regularized channel relative entropy
\begin{align}
D^{\mathrm{reg}}(\mc{M}\|\mc{N}) = \lim_{n\to\infty} \frac{1}{n} \, D(\mc{M}^{\otimes n}\|\mc{N}^{\otimes n}),
\end{align}
where
\begin{align}
D(\mc{M}\|\mc{N}) := \sup_{\rho_{RA}} D\!\Big((\id_R\otimes \mc{M})(\rho_{RA}) \,\Big\|\, (\id_R\otimes \mc{N})(\rho_{RA})\Big).
\end{align}
With this regularization, they established the chain rule
\begin{align}
D\!\Big(\mc{M}(\rho) \,\Big\|\,  \mc{N}(\sigma)\Big)
\;\leq\; D(\rho\|\sigma) + D^{\mathrm{reg}}(\mc{M}\|\mc{N}).
\end{align}

Thus, unlike in the classical case, a quantum chain rule exists only in an asymptotic, many-copy sense, and it is precisely the regularization that guarantees its validity. 

In the present work, we establish a complementary form of chain inequality that holds already at the single-copy level. 

In~\cref{sec:resultsClassicalMethods} we show a single-letter quantum chain rule, some consequeneces and its relation to proviously studied quantities, in~\cref{sec:resultsConditional} we show some general entropy inequalities that conclude in a conditional quantum chain rule and, finally, in~\cref{sec:discussion} we discuss our results and propose some further work on the topic.

\subsection{Notation}

We denote by $D(\cdot\|\cdot)$ the relative entropy, both quantum and classical. We always assume we have some finite dimensional Hilbert space $\H$ and its state space $\S(\H)$. We use Greek letters for quantum states ($\rho,\,\sigma,\,\gamma,\,\omega,\,\tau$) and $p,q$ for classical probability distributions over some finite alphabet $\mc{X}$. Similarly, we use calligraphic letters for quantum maps ($\mc{M},\,\mc{N},\,\mc{R}$) and regular capitalised letters for classical stochastic maps ($M,\,N$). Finally, we denote the Kullblack-Leibler relative entropy between the probability distributions obtained by measuring a POVM $G$ of states $\rho$, $\sigma$ \cite{hayashi99} as $D_G(\rho\|\sigma)=D(P_\rho^G\|P_\sigma^G)$, where $P_\rho^G(j)=\Tr\cwich{G_j\rho}$, and equivalently for $P_\sigma^G$.
For a classical random variable $ X $ with distribution $ P $, we denote the expectation by $\mathbb{E}_{P}[X] := \sum_{x \in \mc{X}} P(x)\, X(x)$ for the discrete set $\mc{X}.$


\section{A Single-Letter Chain Inequality via Measurement-Induced Ensemble Partitions}
\label{sec:resultsClassicalMethods}

In this section we derive a single-letter chain inequality for quantum relative entropy. The proof is based on the DPI: we lift the output states to classical--quantum states carrying an additional branch register and then apply DPI to the partial trace over this register. Conversely, special cases of the inequality recover joint convexity of relative entropy, which is known to imply DPI. In this sense,~\cref{thm:dataProcess1} can be regarded as a chain-rule formulation of data processing.
After proving~\cref{thm:dataProcess1}, we discuss several aspects of the inequality. First, we show how it relates to DPI and to standard properties of relative entropy, including joint convexity. We then explain how the classical--quantum states appearing in the proof arise naturally from the Bayesian joint-state construction. Finally, we derive the semiclassical chain rule in~\cref{coro:difBasis} and compare the state-dependent branch term with $D_A=D^{\mathrm{reg}}$, the optimal state-independent bound in the many-copy chain rule.

The construction is motivated by the classical chain rule. Classically, the loss of distinguishability under two stochastic maps can be compared with an average of branchwise distinguishabilities, where the branches are the point distributions of the input alphabet. This elementary decomposition has no direct quantum analogue: a quantum state does not come equipped with a canonical decomposition into distinguishable points, and different
decompositions of the same density operator are generally incompatible. The main result of this section replaces these point distributions with measurement-induced partitions of quantum states.

Given a POVM $G=\{G_j\}_j$ and a state $\tau$, define
\begin{align}
    P_\tau^G(j)\tau_j
    =
    \sqrt{\tau}\,G_j^{T_\tau}\sqrt{\tau},
    \qquad
    P_\tau^G(j)
    =
    \operatorname{Tr}[G_j\tau],
\end{align}
where $T_\tau$ denotes transposition in a fixed diagonalizing basis of $\tau$. These operators satisfy
\begin{align}
    \tau
    =
    \sum_j P_\tau^G(j)\tau_j,
\end{align}
and therefore form an ensemble decomposition of $\tau$. Thus the unnormalized operators $P_\tau^G(j)\tau_j$ play the role of the classical branches $p_j\delta_j$.


\begin{thm}\label{thm:dataProcess1}
    Let $\rho$, $\sigma$ be quantum states on a finite dimensional Hilbert space $\H_A$. Let $G=\{G_j\}$ be a POVM.  Let $\mc{M},\mc{N}:\B(\H_A)\rightarrow\B(\H_B)$ be completely positive trace preserving (CPTP) maps. Then \begin{align}\label{eq:dataProcessing}
        D(\rho\|\sigma)-D(\mc{M}(\rho)\|\mc{N}(\sigma))\geq-\mathbb{E}_{P_\rho^G}D(\mc{M}(\rho_j)\|\mc{N}(\sigma_j)).
    \end{align}
\end{thm}

\begin{proof}
Fix quantum states $\rho$, $\sigma$ and a POVM $G=\{G_j\}$. Let $\tau\in\{\rho,\sigma\}$. We denote the probability distribution obtained by measuring $\tau$ on $G$ by $P_\tau^G(j)=\Tr\cwich{G_j\tau}$. Moreover, we define the normalised partition elements \begin{equation}
    \tau_j=\frac{\sqrt{\tau}G_j^{T_\tau}\sqrt{\tau}}{\Tr\cwich{G_j\tau}}.
\end{equation} Additionally we have the quantum channels $\mc{M}$, $\mc{N}$ acting on the space of $\rho,\sigma$. 

Let $\tilde{\rho}$, $\tilde{\sigma}$ be the c-q states defined as \begin{equation}
    \tilde{\rho}=\sum_j P_\rho^G(j)\ketbra{j}_R\otimes\mc{M}(\rho_j), \quad \tilde{\sigma}=\sum_j P_\sigma^G(j)\ketbra{j}_R\otimes\mc{N}(\sigma_j),
\end{equation} where $\ketbra{j}_R$ is some classical register on a suitable Hilbert space $\H_R$. Note that $\Tr_R\cwich{\tilde{\rho}}=\mc{M}(\rho)$, $\Tr_R\cwich{\tilde{\sigma}}=\mc{N}(\sigma)$, since $\{P_\rho^G(j)\,\rho_j\}$, $\{P_\sigma^G(j)\,\sigma_j\}$ are partitions. We can write the Umegaki relative entropy of these states as 
\begin{equation}
\begin{split}
D(\tilde{\rho}\|\tilde{\sigma}) &= D\!\bigl(P_\rho^G\|P_\sigma^G\bigr)
+\sum_j P_\rho^G(j)\,D(\mc{M}(\rho_j)\|\mc{N}(\sigma_j))\\
&= D_{G}(\rho\|\sigma)+\sum_{j} P_\rho^G(j)\,D(\mc{M}(\rho_j)\|\mc{N}(\sigma_j)),
\end{split}
\end{equation} 
where the identity follows by evaluating the trace block by block in the common classical decomposition of $\widetilde{\rho}$ and $\widetilde{\sigma}$. On the $j$-th block one uses
$\log(p_j A_j)=\log p_j\,\mathbb{I}+\log A_j$, with
$p_j=P_\rho^G(j)$, $A_j=\mathcal{M}(\rho_j)$, and analogously for $\widetilde{\sigma}$~\cite[Ex. 11.8.8 \& Ex. 11.8.9]{wilde17}. Finally, by the data processing inequality for the partial trace\footnote{We are using only DPI for partial trace channels, but it can be shown to be equivalent to general DPI~\cite[Theorem 11.9.2]{wilde17}.} we obtain
\begin{align}
D(\mc{M}(\rho)\|\mc{N}(\sigma))
= D\!\bigl(\Tr_R[\tilde{\rho}]\,\big\|\,\Tr_R[\tilde{\sigma}]\bigr)
\le D(\tilde{\rho}\|\tilde{\sigma}) = D_G(\rho\|\sigma)+\sum_j P_\rho^G(j)\,D(\mc{M}(\rho_j)\|\mc{N}(\sigma_j)).
\end{align}
Using $D_G(\rho\|\sigma)\leq D(\rho\|\sigma)$, we find
\begin{align}
D(\mc{M}(\rho)\|\mc{N}(\sigma))
\le D(\rho\|\sigma)+\sum_j P_\rho^G(j)\,D(\mc{M}(\rho_j)\|\mc{N}(\sigma_j)),
\end{align}
which rearranges to~\cref{eq:dataProcessing}.

\end{proof}

\begin{remark}
A slightly strengthened inequality relative to~\cref{thm:dataProcess1} is obtained by stopping before the final step $D_G(\rho\|\sigma)\leq D(\rho\|\sigma)$. We choose to express the result with only quantum relative entropies in the Theorem, but it could also be shown as: 
\begin{align}
    D(\mathcal{M}(\rho)\|\mathcal{N}(\sigma))
    \le
    D_{G}(\rho\|\sigma)
    +\mathbb{E}_{P_\rho^G}D(\mc{M}(\rho_j)\|\mc{N}(\sigma_j)).
\end{align}
\end{remark}

\subsection{Corollaries and remarks}

    For a result such as \cref{thm:dataProcess1} one would hope that the data processing inequality can be recovered automatically when letting $\mc M=\mc N$, such as in~\cite{fang20}. Unfortunately, this is not, in general, the case for \cref{thm:dataProcess1}, even if we allow for a choice of the best possible measurement $G$. In this section we explore the relation between \cref{thm:dataProcess1} and DPI, as well as other properties and chain rules of the Umegaki relative entropy.

The inequality in~\cref{eq:dataProcessing} does not reduce to the standard data processing inequality when $\mc{M}=\mc{N}$ for the general case. 
However, if the input states commute, i.e., $[\rho,\sigma]=0$, one can choose a suitable measurement such that~\cref{eq:dataProcessing} specializes to the data processing inequality:

\begin{coro} \label{example:commutingInput}
    Let $\rho$, $\sigma$ be commuting states on some Hilbert space $\H_A$ with common basis $\{\Pi_j\}$ and $\mc{M},\mc{N}:\B(\H_A)\rightarrow \B(\H_B)$ be CPTP maps. Then \begin{align}\label{eq:commutingInput}
        D(\rho\|\sigma)-D(\mc{M}(\rho)\|\mc{N}(\sigma))\geq-\mathbb{E}_{p} D(\mc{M}(\Pi_j)\|\mc{N}(\Pi_j)).
    \end{align}
\end{coro}

\begin{proof}
    In the common basis of $\rho$ and $\sigma$ we write 
    \begin{align}
        \rho=\sum_jp_j\Pi_j,\quad\sigma=\sum_jq_j\Pi_j.
    \end{align} 
    Let $G$ be the projective measurement on the common basis of $\rho,\sigma$, we can write $\rho_j,\sigma_j$ in~\cref{thm:dataProcess1} as \begin{align}
        \rho_j=\frac{(\sqrt{\rho}\Pi_j\sqrt{\rho})^T}{\Tr\cwich{\Pi_j\rho}}=\frac{p_j\Pi_j^T}{p_j}=\Pi_j^T,
    \end{align} and similarly $\sigma_j=\Pi_j^T$. Moreover, $P^G_\rho(j)=p_j$. Therefore~\cref{eq:dataProcessing} becomes 
    \begin{align}
        D(\rho\|\sigma)-D(\mc{M}(\rho)\|\mc{N}(\sigma))\geq-\mathbb{E}_{p}D(\mc{M}(\Pi_j^T)\|\mc{N}(\Pi_j^T)).
    \end{align} Finally, because the transpose is taken in the eigenbasis of $\rho$, in this case the canonical basis, each projector is self-transpose, leaving the final equation as: \begin{align}
        D(\rho\|\sigma)-D(\mc{M}(\rho)\|\mc{N}(\sigma))\geq-\mathbb{E}_{p}D(\mc{M}(\Pi_j)\|\mc{N}(\Pi_j)).
    \end{align} 
\end{proof}

This has the nice property that the right hand side relative entropies are almost independent of $\rho,\sigma$, as they only depend on joint diagonal basis of $\rho$ and $\sigma$. This means, for example, that the equation reduces to Uhlmann's inequality when $\mc{M}=\mc{N}$.

We can use this equation to show joint convexity of relative entropy as well as two extensions:


\begin{coro}\label{coro:jointConvexity}
   \cref{example:commutingInput} with $\rho=\sigma$ is equivalent to the joint convexity of the relative entropy. 
\end{coro}

\begin{proof}
\emph{(Corollary with $\rho=\sigma$ $\Rightarrow$ Joint convexity)} 
Consider the probability distribution $p_{j}$ and collections of states $\{\tau_{j}\}$, $\{\mu_{j}\}$ on a finite-dimensional Hilbert space. Let $\{\Pi_j\}$ be any family of orthogonal rank-one projectors with $\sum_{j} \Pi_{j} = \mathbb{I}$, and set $\rho = \sum_{j} p_{j} \Pi_{j}$.

Let $\mc{M},\mc{N}$ be the classical-quantum channels
\begin{align}\label{eq:MN}
\mathcal M(X)   = \sum_i \tau_i \,\Tr[\Pi_i X], 
\qquad 
\mathcal N(X) = \sum_i \mu_i \,\Tr[\Pi_i X].
\end{align}

By~\cref{example:commutingInput} specialized to $\rho=\sigma$ and to the decomposition of $\rho$ in the basis $\{\Pi_j\}$, we have
\begin{align}\label{eq:example-ineq}
-D(\mathcal M(\rho)\,\|\,\mathcal N(\rho))
\;\ge\;
-\sum_j p_j\, D\bigl(\mathcal M(\Pi_j)\,\|\,\mathcal N(\Pi_j)\bigr).
\end{align}
Using \cref{eq:MN}, this reads
\begin{align}
-D\!\left(\sum_j p_j \tau_j \,\middle\|\, \sum_j p_j \mu_j\right)
\;\ge\;
-\sum_j p_j\, D(\tau_j\|\mu_j),
\end{align}
which is exactly joint convexity of the relative entropy.

\medskip
\emph{(Joint convexity $\Rightarrow$ Corollary with $\rho=\sigma$)}
Conversely, assume joint convexity. Let $\rho$ be any state with spectral decomposition $\rho=\sum_j p_j \Pi_j$ in an orthonormal eigenbasis $\{\Pi_j\}$. For arbitrary channels $\mathcal M,\,\mathcal N$, define
\begin{align}
\tau_{j} = \mathcal M(\Pi_{j}), \qquad \mu_{j} = \mathcal N(\Pi_{j}).
\end{align}
Applying joint convexity to the ensembles $\{p_j,\tau_j\}$ and $\{p_j,\mu_j\}$ gives
\begin{align}
D\!\left(\sum_j p_j \tau_j \,\middle\|\, \sum_j p_j \mu_j\right)
\;\le\;
\sum_j p_j D(\tau_j\|\mu_j).
\end{align}
Since $\sum_j p_j \tau_j=\mathcal M(\rho)$ and $\sum_j p_j \mu_j=\mathcal N(\rho)$, this is precisely
\begin{align}
-D(\mathcal M(\rho)\,\|\,\mathcal N(\rho))
\;\ge\;
-\sum_j p_j\, D\bigl(\mathcal M(\Pi_j)\,\|\,\mathcal N(\Pi_j)\bigr),
\end{align}
which is the $\rho=\sigma$ instance of~\cref{example:commutingInput}. Hence the two statements are equivalent.
\end{proof}

\begin{remark}\label{rem:DPIproof}
    It was shown in~\cite{ruskai07} by Ruskai that joint convexity is sufficient to establish both the data processing inequality~\cite{lindblad75} and the strong subadditivity of quantum relative entropy~\cite{lieb73}. By contrast, the main-text proof of~\cref{thm:dataProcess1} works by lifting the problem to the c-q states $\tilde{\rho}$ and $\tilde{\sigma}$, applying the data processing inequality for the partial trace, and then using monotonicity under the measurement $G$ to pass from $D_G(\rho\|\sigma)$ to $D(\rho\|\sigma)$ \cite{hayashi06,mullerLennert13}. The appendix proof instead derives~\cref{thm:dataProcess1} from the classical chain rule via Jensen's inequality, Uhlmann's inequality, and the asymptotic equipartition property of Hiai and Petz~\cite{hiai91}. Consequently,~\cref{coro:jointConvexity} can in turn be employed to rederive these properties from~\cref{thm:dataProcess1}.
\end{remark}

Notice that~\cref{example:commutingInput} also provide an extension to the joint convexity when $\rho = \sigma$ is relaxed.


\begin{coro}\label{coro:ensembles}
    Let $\{\tau_i\},\{\mu_i\}\subseteq\S(\H_B)$ be collections of states and let $p=\{p_j\}, q=\{q_j\}$  probability distributions. Then 
    \begin{align}
            D(p\|q)-D\left(\sum_jp_j\tau_j\right|\!\left|\sum_jq_j\mu_j\right)\geq  -\sum_jp_jD(\tau_j\|\mu_j).
    \end{align}
\end{coro}

\begin{proof}
Let $\{\Pi_j\}$ be a family of mutually orthogonal rank-one projectors with $\sum_j \Pi_j=\I$ and define
$\rho=\sum_j p_j \Pi_j$, $\sigma=\sum_j q_j \Pi_j$, where $p,q$ are probability distributions.
Note that $\rho$ and $\sigma$ commute, hence $D(\rho\|\sigma) = D(p\|q)$.   Define the channels 
\begin{align}
        \mc{M}(x)=\sum_i\tau_i\Tr\cwich{\Pi_i x}\quad \text{and}\quad \mc{N}(x)=\sum_i\mu_i\Tr\cwich{\Pi_i x}.
    \end{align}
 Applying inequality~\cref{eq:commutingInput} to $\rho,\sigma$ with these maps gives
\begin{align}
   D(p \| q ) - D\!\left(\mc{M}(\rho)\,\middle\|\,\mc{N}(\sigma)\right)  \geq - \sum_j p_j\, D\!\big(\mc{M}(\Pi_j)\,\|\,\mc{N}(\Pi_j)\big).
\end{align}
By construction, $\mc{M}(\rho) = \sum_j p_j \tau_j$
 $\mc{N}(\sigma) = \sum_j q_j \mu_j$, 
 $\mc{M}(\Pi_j) = \tau_j$ and $\mc{N}(\Pi_j) = \mu_j$.

Substituting these identities gives the claim.
\end{proof}

Finally, in the general case where the states and channels are arbitrary, 
we obtain a semiclassical variant of the chain rule. 
It relates the relative entropy of the eigenvalue distributions of $\rho$ and $\sigma$ 
to the divergence of their images under the maps $\mc{M},\mc{N}$ and the action of these maps on the corresponding eigenprojectors.

\begin{coro}\label{coro:difBasis}
Let $\rho,\sigma \in \S(\H_A)$ be quantum states with spectral decompositions
$\rho = \sum_j p_j \Pi_j$ and $\sigma = \sum_j q_j \tilde{\Pi}_j$,
where $p=\{p_j\}$ and $q=\{q_j\}$ are the eigenvalue distributions and 
$\{\Pi_j\},\{\tilde{\Pi}_j\}$ are the associated rank-one eigenprojectors.  
Let $\mc{M},\mc{N}:\B(\H_A)\to\B(\H_B)$ be CPTP maps. 
Then
\begin{align}
   D(p\|q) - D\!\bigl(\mc{M}(\rho)\,\|\,\mc{N}(\sigma)\bigr)
   \;\;\geq\; -\,\mathbb{E}_{p} D\!\bigl(\mc{M}(\Pi_j)\,\|\,\mc{N}(\tilde{\Pi}_j)\bigr).
\end{align}
\end{coro}

\begin{proof}
First observe that, by the spectral decomposition of $\sigma$,  
\begin{align}
   \mc N(\sigma) = \sum_j q_j \mc N(\tilde{\Pi}_j).
\end{align}
Choose a unitary channel $\mc U$ mapping the eigenbasis of $\sigma$ to that of $\rho$, i.e.\ 
$\mc{U}(\tilde{\Pi}_j)=\Pi_j$\footnote{Note that such a unitary is not unique, and different choices may lead to different values of 
$D(p\|q)$ and of the average 
$\mathbb{E}_{p}[D(\mc{M}(\Pi_i)\|\mc{N}(\tilde{\Pi}_i))]$; this dependence is expanded upon in~\cref{rem:ordering}.}.   

Define $\mc F=\mc N\circ \mc U$ and $\sigma'=\sum_j q_j \Pi_j$.
Then $\mc N(\sigma)=\mc F(\sigma')$, and note that $\rho$ and $\sigma'$ commute.  

Applying the inequality~\cref{eq:commutingInput} to $\rho,\sigma'$ and the channels 
$\mc M,\mc F$ yields
\begin{align}
   D(p\|q) - D(\mc M(\rho)\|\mc N(\sigma))
   \;\;\geq\; - \sum_j p_j D(\mc M(\Pi_j)\|\mc N(\tilde{\Pi}_j)),
\end{align}
as claimed.
\end{proof}

\begin{remark}\label{rem:ordering}
The inequalities in~\cref{coro:ensembles} and~\cref{coro:difBasis} depend on an arbitrary choice of ordering: of the ensemble elements in~\cref{coro:ensembles} and of the eigenbasis in~\cref{coro:difBasis}. In both cases the second term on the left-hand side,
\begin{align*}
   D\!\left(\sum_j p_j \tau_j \,\middle\|\, \sum_j q_j \mu_j\right) 
   \quad\text{and}\quad
   D(\mc M(\rho)\|\mc N(\sigma)),
\end{align*}
is invariant under reordering. This is also the term of primary interest. 
A possible refinement would be to optimize the bound over all $n!$ permutations, where $n$ is the size of the ensemble or, equivalently, the Hilbert space dimension.
\end{remark}

Moreover, notice that~\cref{coro:difBasis} can be recovered directly from~\cref{coro:ensembles}: 
choosing $\{p_i,\tau_i\}=\{p_i,\mc M(\Pi_i)\}$ and 
$\{q_i,\mu_i\}=\{q_i,\mc N(\tilde{\Pi}_i)\}$ in~\cref{coro:ensembles} yields 
\cref{coro:difBasis} immediately.

\subsection{Bayesian interpretation of the c-q states}\label{subsec:bayesian_interpretation}

In this subsection we explain how the classical--quantum states used in the proof of~\cref{thm:dataProcess1} arise from the Bayesian joint-state construction~\cite{fuchs01,leifer13,horsman17}. This identifies the branch register with the outcome of a measurement on a reference system. 

This perspective also connects the direct c-q proof in the main text with the alternative classical proof in~\cref{app:altProof}. In the latter proof, both the reference and output systems are measured, and the use of DPI is replaced by the classical chain rule for relative entropy; the quantum statement is then recovered through a measurement-to-quantum limiting argument.

We begin by recalling the relevant constructions involving ensemble partitions. For any state $\tau$ and POVM $G=\{G_j\}$, define
\begin{align}
   \tau_j &= \frac{\sqrt{\tau}\, G_j^{T_\tau}\, \sqrt{\tau}}{\Tr[G_j \tau]}, 
   & P_\tau^G(j) &= \Tr[G_j \tau],
\end{align}
where $T_{\tau}$ denotes the transpose in a diagonal basis of $\tau$ and $P_\tau^G(j)=\Tr\cwich{G_j\tau}$ is the probability distribution associated with measuring $G$  on $\tau$.

Given a state $\tau$ and a quantum channel $\varepsilon$, we also consider the bipartite state
\begin{align}
   \omega_\tau^\varepsilon = (\id \otimes \varepsilon)(\kketbbra{\tau}),
\end{align}
where $\kket{\tau} = \sum_k \sqrt{t_k}\,\ket{t_k,t_k}$ is the canonical purification of $\tau$.

These constructions were first introduced by Fuchs~\cite{fuchs01} in the study of quantum conditional probabilities and later developed by Leifer and collaborators~\cite{leifer06,leifer13} as part of a program to generalize Bayes' theorem to the quantum setting. In this framework, the unnormalised $\tau_j$ are understood as ensemble partitions of $\tau$, providing the analogue of the Bayesian update rule for classical distributions under measurement.
More recently, Parzygnat and Fullwood~\cite{parzygnat23} reformulated these constructions in a structural framework, treating states and channels in terms of their compositional properties rather than their specific operator representations, thereby giving structural interpretations of partitions as quantum conditionals. The bipartite state $\omega_\tau^\varepsilon$ was proposed in~\cite{leifer13} as a candidate for the quantum analogue of the joint distribution $p(y|x)p(x)$, although later work~\cite{horsman17,fullwood22,lie23} has argued that, while it shares some features with joint classical distributions, it may not fully capture the operational or informational content of such objects. Building on these constructions, we obtain~\cref{thm:dataProcess1}.

In this setting, we can see how the measurement $G$ naturally appears as part of a physical operation on the joint states.

Let $\omega_\tau^{\mathcal{E}}$ be the Bayesian joint state associated with a state $\tau$ and a channel $\mathcal{E}$. Let $G=\{G_j\}_j$ be a POVM on the reference system of this joint state. We denote by
\begin{align}
    \Phi_G(X_R)
    =
    \sum_j \Tr[G_jX_R]\ketbra{j}_J
\end{align}
the usual quantum-to-classical measurement channel associated with $G$, and define its extension to the joint system by
\begin{align}
    \mathcal{Q}_G
    =
    \Phi_G\otimes\id_B .
\end{align}
Equivalently, for an operator $X_{RB}$ on the joint system,
\begin{align}
    \mathcal{Q}_G(X_{RB})
    =
    \sum_j
    \ketbra{j}_J
    \otimes
    \Tr_R\!\left[
        (G_j\otimes \mathbb{I}_B)X_{RB}
    \right].
    \label{eq:QG_definition_main}
\end{align}
Since the trace is over the measured system, this can also be written as
\begin{align}
    \mathcal{Q}_G(X_{RB})
    =
    \sum_j
    \ketbra{j}_J
    \otimes
    \Tr_R\!\left[
        (\sqrt{G_j}\otimes \mathbb{I}_B)
        X_{RB}
        (\sqrt{G_j}\otimes \mathbb{I}_B)
    \right],
\end{align}
which makes complete positivity manifest.

By applying this channel to $\omega^{\mc{E}}_\rho$ we can recover the c-q state in the proof of \cref{thm:dataProcess1} since
\begin{align}
    \Tr_R\!\left[
        (G_j\otimes \mathbb{I}_B)
        \omega_\tau^{\mathcal{E}}
    \right]
    =
    \mathcal{E}\!\left(
        \sqrt{\tau}\,G_j^{T_\tau}\sqrt{\tau}
    \right).
    \label{eq:bayesian_block_identity_main}
\end{align}
Therefore, with
\begin{align}
    P_\tau^G(j)=\Tr[G_j\tau],
    \qquad
    \tau_j
    =
    \frac{
        \sqrt{\tau}\,G_j^{T_\tau}\sqrt{\tau}
    }{
        P_\tau^G(j)
    },
\end{align}
for all outcomes with $P_\tau^G(j)>0$, we obtain
\begin{align}
    \mathcal{Q}_G(\omega_\tau^{\mathcal{E}})
    =
    \sum_j
    P_\tau^G(j)\,
    \ketbra{j}_J
    \otimes
    \mathcal{E}(\tau_j).
    \label{eq:bayesian_to_cq_main}
\end{align}
Thus the c-q states used in the proof of the chain inequality are precisely the Bayesian joint states after applying the measurement channel associated with $G$ to the reference system.

In particular, for the pairs $(\rho,\mathcal{M})$ and
$(\sigma,\mathcal{N})$, we obtain
\begin{align}
    \widetilde{\rho}
    =
    \mathcal{Q}_G(\omega_\rho^{\mathcal{M}})
    =
    \sum_j
    P_\rho^G(j)
    \ketbra{j}_J
    \otimes
    \mathcal{M}(\rho_j),
\end{align}
and
\begin{align}
    \widetilde{\sigma}
    =
    \mathcal{Q}_G(\omega_\sigma^{\mathcal{N}})
    =
    \sum_j
    P_\sigma^G(j)
    \ketbra{j}_J
    \otimes
    \mathcal{N}(\sigma_j).
\end{align}

This also clarifies the relation with the alternative proof in \cref{app:altProof}. If one additionally measures the output system with a POVM $F=\{F_i\}_i$, with measurement channel
\begin{align}
    \Phi_F(Y_B)
    =
    \sum_i
    \Tr[F_iY_B]\ketbra{i}_I,
\end{align}
then the fully measured state defines the classical joint distribution
\begin{align}
    P_{\tau,\mathcal{E}}^{G,F}(j,i)
    =
    P_\tau^G(j)\,
    \Tr[F_i\mathcal{E}(\tau_j)].
\end{align}
Equivalently,
\begin{align}
    P_{\tau,\mathcal{E}}^{G,F}(j,i)
    =
    \Tr\!\left[
        (G_j\otimes F_i)
        \omega_\tau^{\mathcal{E}}
    \right].
\end{align}
Hence the two proofs are connected by the hierarchy
\begin{align}
    \omega_\tau^{\mathcal{E}}
    \quad
    \xrightarrow{\ \mathcal{Q}_G\ }
    \quad
    \sum_j
    P_\tau^G(j)\ketbra{j}_J\otimes\mathcal{E}(\tau_j)
    \quad
    \xrightarrow{\ \id_J\otimes\Phi_F\ }
    \quad
    P_{\tau,\mathcal{E}}^{G,F}(j,i).
\end{align}
The proof in the main text keeps the output system quantum;it then uses DPI for the partial trace on the resulting c-q
states. The alternative proof measures the output system as well, so that the argument becomes fully classical and uses the classical chain rule for relative entropy. The quantum statement is then recovered through the
measurement-to-quantum limiting argument described in~\cref{app:altProof}.

\subsection{Comparison with the Optimal State-Independent Bound}
\label{sec:DregComp}

We now compare the branch term in!\cref{thm:dataProcess1} with the state-independent constant that appears in the many-copy chain rule. The relevant benchmark is the amortized channel divergence $D_A(\mc{M}\|\mc{N})$. By definition, $D_A(\mc{M}\|\mc{N})$ is the smallest constant $C(\mc{M},\mc{N})$, independent of the input states, such that
\begin{align}
    D\!\left(
        (\mc{M})(\rho)
        \middle\|
        (\mc{N})(\sigma)
    \right)
    \leq
    D(\rho\|\sigma)
    +
    C(\mc{M},\mc{N})
\end{align}
holds for all finite-dimensional ancillary systems $R$ and all states $\rho ,\sigma$. Fang~\emph{et al.}~\cite{fang20} showed that this amortized quantity  coincides with the regularized channel divergence,
\begin{align}
    D_A(\mc{M}\|\mc{N})
    =
    D^{\mathrm{reg}}(\mc{M}\|\mc{N}).
\end{align}
Thus $D^{\mathrm{reg}}(\mc{M}\|\mc{N})$ is the optimal state-independent correction term in the asymptotic channel chain rule.
By contrast, \cref{thm:dataProcess1} replaces the state-independent correction term $D^{\mathrm{reg}}(\mc{M}\|\mc{N})$ by the state- and
measurement-dependent branch average
\begin{align}
    \mathbb{E}_{P_{\rho}^{G}}
    D\!\left(\mc{M}(\rho_{j})\middle\|\mc{N}(\sigma_{j})\right).
\end{align}
These two quantities are not directly comparable in general: the former is an intrinsic channel quantity, optimized independently of the input states, whereas the latter depends on the chosen states $(\rho,\sigma)$ and on the
POVM-induced partitions. 

The following lemma gives one useful comparison. It bounds the branch average by $D^{\mathrm{reg}}(\mc{M}\|\mc{N})$ plus a
mismatch term measuring how far the conditional branches $\rho_j$ and $\sigma_j$ remain apart after applying $\mc{N}$.
\begin{align}
    D_{\max}(X\|Y) :=
    \inf\left\{\lambda\in\mathbb{R}:X\leq e^{\lambda}Y\right\},
\end{align}
with the convention that $D_{\max}(X\|Y)=+\infty$ if no finite $\lambda$ satisfies the operator inequality.

\begin{lemma}[Comparison with $D^{\mathrm{reg}}$]
\label{lemma:DmaxComparisonDreg}
Let $\rho,\sigma\in\S(\H_{A})$, let $G=\{G_{j}\}_{j}$ be a POVM, and let
$\rho_{j},\sigma_{j}$ be the conditional states appearing in
\cref{thm:dataProcess1}. Then
\begin{align}
    &\mathbb{E}_{P_{\rho}^{G}}
    D\!\left(\mc{M}(\rho_{j})\middle\|\mc{N}(\sigma_{j})\right)\leq
    D^{\mathrm{reg}}(\mc{M}\|\mc{N})+
    \mathbb{E}_{P_{\rho}^{G}}D_{\max}\!\left(\mc{N}(\rho_{j})\middle\|\mc{N}(\sigma_{j})\right).
    \label{eq:DmaxComparisonDreg}
\end{align}
\end{lemma}

\begin{proof}
Fix a branch $j$ with $P_{\rho}^{G}(j)>0$. If
\begin{align}
    D_{\max}\!\left(\mc{N}(\rho_{j})\middle\|\mc{N}(\sigma_{j})\right)=+\infty,
\end{align}
there is nothing to prove for this branch. Otherwise, since we work in finite
dimension, the infimum in $D_{\max}$ is attained whenever it is finite. Hence
\begin{align}
    \mc{N}(\sigma_{j})
    \geq
    \exp\!\left[-D_{\max}\!\left(\mc{N}(\rho_{j})\middle\|\mc{N}(\sigma_{j})\right)\right]
    \mc{N}(\rho_{j}).
\end{align}
By operator monotonicity of the logarithm,
\begin{align}
    -\log\mc{N}(\sigma_{j})\leq-\log\mc{N}(\rho_{j})+D_{\max}\!\left(\mc{N}(\rho_{j})\middle\|\mc{N}(\sigma_{j})\right)\I .
\end{align}
Multiplying by the positive operator $\mc{M}(\rho_{j})$, taking the trace,
and using $\Tr[\mc{M}(\rho_{j})]=1$, we obtain
\begin{align}
    &D\!\left(\mc{M}(\rho_{j})\middle\|\mc{N}(\sigma_{j})\right)
    \leq
    D\!\left(\mc{M}(\rho_{j})\middle\|\mc{N}(\rho_{j})\right)
    +
    D_{\max}\!\left(\mc{N}(\rho_{j})\middle\|\mc{N}(\sigma_{j})\right).
    \label{eq:branchDmaxComparison}
\end{align}
The first term on the right-hand side is bounded by the channel
relative entropy:
\begin{align}
    D\!\left(\mc{M}(\rho_{j})\middle\|\mc{N}(\rho_{j})\right)\leq
    D(\mc{M}\|\mc{N})  \leq D^{\mathrm{reg}}(\mc{M}\|\mc{N}),
\end{align}
Combining this with~\cref{eq:branchDmaxComparison} gives
\begin{align}
    D\!\left(\mc{M}(\rho_{j})\middle\|\mc{N}(\sigma_{j})\right)
    \leq
    D^{\mathrm{reg}}(\mc{M}\|\mc{N})
    +
    D_{\max}\!\left(\mc{N}(\rho_{j})\middle\|\mc{N}(\sigma_{j})
    \right).
\end{align}
Averaging over $P_{\rho}^{G}$ gives \cref{eq:DmaxComparisonDreg}.
\end{proof}

\begin{remark}
    The bound can be loosened to remove the channel dependence using DPI for the max relative entropy to obtain \begin{align}
    &\mathbb{E}_{P_{\rho}^{G}}
    D\!\left(\mc{M}(\rho_{j})\middle\|\mc{N}(\sigma_{j})\right)\leq
    D^{\mathrm{reg}}(\mc{M}\|\mc{N})+
    \mathbb{E}_{P_{\rho}^{G}}D_{\max}\!\left(\rho_{j}\middle\|\sigma_{j}\right).
\end{align}
\end{remark}

While there will exist states and measurements for which the distributed average is greater than the regularized entropy, that is not generally the case. The following provide classes of tuples $(\rho,\sigma,\mc{M},\mc{N},G)$ such that \begin{equation}
    \mathbb{E}_{P^G_{\rho}}D(\mc{M}(\rho_j)\|\mc{N}(\sigma_j))< D^{\mathrm{reg}}(\mc{M}\|\mc{N}).
\end{equation} In particular, we provide three cases where is is less: and academic example as well as the more general cases of commuting and pure states. 

\begin{example}
    Let $d=2$, $\rho= \ketbra{0}$, $\sigma = (1-\varepsilon)\ketbra{+}+\varepsilon\ketbra{-}$, $\mc{M}(x)=\Tr(x)\ketbra{-}$ and $\mc{N}=\mc{E}_\sigma$, where $\mc{E}_\sigma$ denotes the pinching map on the basis of $\sigma$ \cite{hayashi99}. Let $G=\set{\ketbra{0},\ketbra{1}}$ Then \begin{align}
        \mathbb{E}_pD\!\swich{\mc{M}(\rho_i)\|\mc{N}(\sigma_i)}&=D\!\swich{\mc{M}(\ketbra{0})\|\mc{N}((\sqrt{\sigma} \ketbra{0}\sqrt{\sigma})/\expval{\sigma}{0})}\\
        &=D(\ketbra{-}\| (1-\varepsilon)\ketbra{+}+\varepsilon\ketbra{-})=-\log \varepsilon.
    \end{align} Meanwhile \begin{align}
    D^{\mathrm{reg}}(\mc{M}\|\mc{N})&=\lim_{n\to\infty}\frac{1}{n}D(\mc{M}^{\otimes n}\|\mc{N}^{\otimes n})\geq\lim_{n\to\infty}\frac{1}{n}n\max_{\phi\in\S(\H)}D(\mc{M}(\phi)\|\mc{N}(\phi))\\
    &\geq D(\mc{M}(\ketbra{+})\|\mc{N}(\ketbra{+}))=D(\ketbra{-}\|\ketbra{+})=+\infty.
    \end{align} Regardless of $\varepsilon\in(0,1)$ we see that $\mathbb{E}_pD\!\swich{\mc{M}(\rho_i)\|\mc{N}(\sigma_i)}$ is finite, while $D^{\mathrm{reg}}(\mc{M}\|\mc{N})$ is not. 
\end{example}

\begin{example}[Commuting inputs]
The mismatch term in \cref{lemma:DmaxComparisonDreg} vanishes in the
commuting case. Let
\begin{align}
    \rho=
    \sum_{k}p_{k}\Pi_{k},\qquad\sigma=
    \sum_{k}q_{k}\Pi_{k}
\end{align}
be commuting states, and choose the common eigenbasis POVM
\begin{align}
    G=\{\Pi_{k}\}_{k}.
\end{align}
Assume $q_{k}>0$ whenever $p_{k}>0$. Then, for every branch with
$p_{k}>0$,
\begin{align}
    \rho_{k}=
    \frac{\sqrt{\rho}\Pi_{k}\sqrt{\rho}}{\Tr[\Pi_{k}\rho]}=\Pi_{k},
    \qquad
    \sigma_{k}=
    \frac{\sqrt{\sigma}\Pi_{k}\sqrt{\sigma}}{\Tr[\Pi_{k}\sigma]}=
    \Pi_{k}.
\end{align}
Therefore
\begin{align}
    D_{\max}\!\left(\mc{N}(\rho_{k})\middle\|\mc{N}(\sigma_{k})\right)=
    D_{\max}\!\left(\mc{N}(\Pi_{k})\middle\|\mc{N}(\Pi_{k})\right)=0.
\end{align}
Hence
\begin{align}
    \mathbb{E}_{P_{\rho}^{G}}
    D_{\max}\!\left(\mc{N}(\rho_{k})\middle\|\mc{N}(\sigma_{k})\right)=0.
\end{align}
By \cref{lemma:DmaxComparisonDreg}, we obtain
\begin{align}
    \mathbb{E}_{P_{\rho}^{G}}
    D\!\left(\mc{M}(\rho_{k})\middle\|\mc{N}(\sigma_{k})\right)\leq D^{\mathrm{reg}}(\mc{M}\|\mc{N}).
\end{align}
Thus, for commuting inputs and the common eigenbasis measurement, the
state-dependent term in~\cref{thm:dataProcess1} is no larger than the
regularized state-independent channel divergence.
\end{example}

The commuting case shows a gap between the bounds for commuting states with almost all pairs of channels. This gap and the continuity of the relative entropy (for non pure states) implies that this is also true from almost commuting states, that is pairs of states such that $\fmig\norm{[\rho,\sigma]}_1= \varepsilon\ll1$.

\begin{remark}[Pure states]
For pure states $\rho$ and $\sigma$ the inequality 
\begin{align}
D(\rho\|\sigma) -D(\mathcal{M}(\rho)\|\mathcal{N}(\sigma))\ge -\mathbb{E}_{P^G_{\rho}}D(\mc{M}(\rho_j)\|\mc{N}(\sigma_j))
\end{align}      
trivialises to $D(\rho\|\sigma)\ge 0$ for all $G$, as $\ketbra{\psi}_j=\ketbra{\psi}$ for all $\ket{\psi}$ and $G_j$.
\end{remark}

\section{Conditional Chain Rule via Entropy Bounds}\label{sec:resultsConditional}

In this section we establish a conditional chain rule (see~\cref{coro:chainRule1}). The result follows as a corollary of~\cref{thm:generalEntropyIneq}, which provides a general entropy inequality. 
By choosing suitable auxiliary states, we derive both the chain rule and~\cref{coro:2channelDPI}, thereby extending the framework of Sutter \emph{et al.}~\cite{sutter16}  to different channels.

Unlike in the setting of Sutter \emph{et al.} (see eq.~\cref{eq:universalBound}), our framework does not currently allow optimization over all possible measurements. We believe this reflects a technical limitation of our proof method rather than a fundamental restriction of the approach itself.
The central tool in our approach is a recovery map that generalizes the standard Petz construction. 

Whereas the Petz map is defined with respect to a single reference state, our formulation involves two distinct reference states. We define the resulting \emph{twisted recovery map}, for $\alpha = (1 + \ii t)/2$ as
\begin{align}
    \mc{R}^\alpha_{\gamma,\sigma,\mc{M}}(X)=J^\alpha_\sigma\circ\mc{M}^\dagger\circ J^{-\alpha}_{\gamma}(X)
\end{align}    
with $J_\sigma^\alpha(X) = \sigma^\alpha X \sigma^{\alpha^*}$. 

Analogously to the standard Petz construction, the pair $(\gamma,\sigma)$ plays the role of \emph{reference states} that fix
the normalization convention on the input and output side of the recovery.
For $\gamma=\mathcal{M}(\sigma)$  this reduces to the rotated Petz map $\mathcal R^{\alpha}_{\sigma,\mathcal{M}}$.

The map $\mathcal R^\alpha_{\gamma,\sigma,\mathcal M}$ is, like the Petz map, uniquely defined in the support of $\gamma$ and completely positive but not, in general, trace-non-increasing. 
This follows directly from the fact that the reference states $(\gamma,\sigma)$ are chosen independently: unless they satisfy $\gamma=\mathcal{M}(\sigma)$, the rescaling induced by $J_{\gamma}^{-\alpha}$
and the normalization by $J_{\sigma}^{\alpha}$ are mismatched, and the overall trace is dependent on the input, even in the support of $\gamma$.
The rotated Petz case, recovered for $(\gamma,\sigma)=(\mathcal{M}(\sigma),\sigma)$, is the unique configuration that restores trace preservation in the suport of $\gamma$. 

As in the standard Petz construction, we consider the rotated family $\mathcal R^{\alpha}_{\gamma,\sigma,\mathcal M}$, weighted by the real function $\beta_0(t) $, and define the averaged map
\begin{align}
\overline{\mathcal R}_{\gamma,\sigma,\mathcal M}
=\int \beta_0(t)\,\mathcal R^{\alpha}_{\gamma,\sigma,\mathcal M}\,dt,
\qquad 
    \beta_0(t) = \tfrac{\pi}{2}\bigl(\cosh(\pi t)+1\bigr)^{-1}.
\end{align}
Averaging preserves complete positivity but does not generally provide any information on the trace.
The averaged map is trace-non-increasing only when the effective rescaling operator obtained from the integral is bounded from above by the identity.
This condition holds automatically for the matched Petz reference states but requires additional, channel-dependent relations for arbitrary pairs $(\gamma,\sigma)$.

The situations in which the map becomes trace-non-increasing naturally identify classes of states and channels for which probabilistic single-shot recovery is achievable.
In this sense, the recovery map functions both as a mathematical device for proving entropy inequalities and as a diagnostic for detecting operationally meaningful recovery regimes.

Theorem \ref{thm:generalEntropyIneq} formalizes the general entropy inequality underlying this construction. Its proof adapts techniques from Kwon and Kim~\cite{kwon19,sutter16} and relies on operator inequalities such as the Peierls-Bogoliubov inequality.
 From there, we derive specialized forms by making informed choices for the auxiliary states, eventually leading to the chain rule as a corollary.


\begin{thm}\label{thm:generalEntropyIneq}

    Let $\mc{M}:\B(\H_A)\rightarrow\B(\H_B)$ be a quantum channel, let
    $\rho,\sigma \in \mathcal S(\H_A)$ be states on the input space, and let
    $\gamma,\omega \in \B(\H_B)$ be positive operators on the output space.
    Assume that
    \begin{align}
        \supp(\mc{M}(\rho))\subseteq\supp(\gamma).
    \end{align}
    Then
    \begin{align}
    D(\rho\|\sigma)-D(\mc{M}(\rho)\|\gamma)+D(\mc{M}(\rho)\|\omega)
    \geq
    D_\Pi(\rho\|\overline{\mc{R}}_{\gamma,\sigma,\mc{M}}(\omega)),
    \end{align}
    with $\Pi=\{\Pi_i\}$ the rank-one POVM defined by a diagonal basis of $\rho$ and
    \begin{align}
    \overline{\mc{R}}_{\gamma,\sigma,\mc{M}}=
    \int_{-\infty}^{+\infty}d\beta_0(t)\,
    J^\alpha_\sigma\circ\mc{M}^\dagger\circ J^{-\alpha}_{\gamma},
    \end{align}
    where $\beta_0(t)=\frac{\pi}{2}(\cosh(\pi t)+1)^{-1}$ and
    $\alpha=\frac{1+\ii t}{2}$, and the negative powers of $\gamma$ are understood in the Moore--Penrose sense.
\end{thm}

\begin{proof}

To avoid domain issues caused by non-full-rank operators, we use the following
regularization.  If $X\geq0$ is a nonzero positive semidefinite operator on a finite-dimensional Hilbert space $\mathcal K$, define
\begin{align}
    X_\varepsilon
    :=
    \frac{\Tr[X]}{\Tr[X+\varepsilon \I_{\mathcal K}]}
    \left(X+\varepsilon \I_{\mathcal K}\right),
    \qquad \varepsilon>0 .
\end{align}
Then $X_\varepsilon>0$, $\Tr[X_\varepsilon]=\Tr[X]$, and
$X_{\varepsilon} \to X $ in operator norm as $\varepsilon \downarrow 0$.Throughout, negative powers of positive
semidefinite operators are understood in the Moore--Penrose sense.

This regularization allows us to extend expressions involving non-full-rank $\sigma$ and $\omega$ directly by continuity. 

The case of a non-full-rank $\gamma$ requires a slightly different
regularization, because inverse powers are not continuous when an eigenvalue approaches zero. We therefore do not regularize $\gamma$ directly. Instead, we pass to its Moore--Penrose inverse.

Let $\xi$ denote the Moore--Penrose inverse of $\gamma$, namely
\begin{align}
    \xi=\gamma^{-1}
    \quad \text{on} \operatorname{supp}\gamma,
    \qquad
    \xi=0
    \quad \text{on }(\operatorname{kerr}\gamma).
\end{align}
Since $\xi_{\varepsilon}\to \xi$ in operator norm and only positive powers of $\xi_{\varepsilon}$ are involved, we have that 
\begin{align}
    J^{\alpha}_{\xi_{\varepsilon}}
    \xrightarrow{\varepsilon\to 0}
    J^{\alpha}_{\xi}
    =
    J^{-\alpha}_{\gamma}.
\end{align}
Let $\rho = \sum_{i} p_{i} \Pi_{i}$ be a rank-1 decomposition of $\rho$. For $i$ such that $p_{i}>0$. 
Consider the quantity $\Tr\cwich{p_{i}^{-1}\Pi_{i,\varepsilon}(J^\alpha_{\sigma_{\varepsilon}}\circ\mc{M}^\dagger\circ J^{\alpha}_{{\xi_{\varepsilon}}})({\omega_{\varepsilon}})}$. Let $U$ be the unitary that dilates the map $\mc{M}^{\dagger}(\cdot)=U^{\dagger}\cdot\otimes\I_{E} U$. For notational simplicity we omit the ancillary identity. Let $\alpha=\frac{1+\ii t}{2}$ then
\begin{align}
    \begin{split}
        \Tr&\cwich{p_{i}^{-1}\Pi_{i,\varepsilon}(J^\alpha_{\sigma_{\varepsilon}}\circ\mc{M}^\dagger\circ J^{\alpha}_{{\xi_{\varepsilon}}})({\omega_{\varepsilon}})}\\
    =&\Tr\cwich{\Pi_{i,\varepsilon} (p_{i}\I)^{-1}{\sigma_{\varepsilon}^\alpha} U^\dagger{\xi_{\varepsilon}^\alpha} UU^\dagger{\omega_{\varepsilon}} UU^\dagger{\xi_{\varepsilon}^{\alpha^*}}U{\sigma_{\varepsilon}^{\alpha^*}}}\\
    =&\Tr\cwich{e^{\log \Pi_{i,\varepsilon}}e^{-\log p_{i}\I}e^{\alpha\log{\sigma_{\varepsilon}}}e^{\alpha U^\dagger\log{\xi_{\varepsilon}} U}e^{U^\dagger\log{\omega_{\varepsilon}} U}e^{\alpha^* U^\dagger\log{\xi_{\varepsilon}} U}e^{\alpha^*\log{\sigma_{\varepsilon}}}}\\
    =&\norm{e^{\alpha\log \Pi_{i,\varepsilon}}e^{-\alpha\log p_{i}\I}e^{\alpha\log{\sigma_{\varepsilon}}}e^{\alpha U^\dagger\log{\xi_{\varepsilon}} U}e^{\alpha U^\dagger\log{\omega_{\varepsilon}} U}}_2^2\\
    =&\norm{\prod_{k=1}^5e^{2\alpha H_k}}_2^2.
    \end{split}
\end{align}
 With $H_1=\fmig\log \Pi_{i,\varepsilon}$, $H_2=-\fmig\log p_{i} \I$, $H_3=\fmig\log{\sigma_{\varepsilon}}$, $H_4=\fmig U^\dagger\log{\xi_{\varepsilon}} U$ and $H_5=\fmig U^\dagger\log{\omega_{\varepsilon}} U$.
 
 We now integrate over $t$ with respect to the weight $d\beta_0(t)$, apply the logarithm on both sides of the equality and average over the spectrum of $\rho$. The left hand side, in the limit $\varepsilon\to 0$, becomes
\begin{align}
\begin{split}
&\lim_{\varepsilon \to 0}\sum_{i} p_{i} \log \int_{-\infty}^{\infty}d\beta_0(t)\Tr\cwich{p_{i}^{-1}\Pi_{i,\varepsilon}(J^\alpha_{\sigma_{\varepsilon}}\circ\mc{M}^\dagger\circ J^{\alpha}_{{\xi_{\varepsilon}}})({\omega_{\varepsilon}})}\\ 
&=\sum_{i} p_{i} \log \int_{-\infty}^{\infty}d\beta_0(t)\Tr\cwich{p_{i}^{-1}\Pi_{i}(J^\alpha_{\sigma}\circ\mc{M}^\dagger\circ J^{-\alpha}_{{\gamma}})({\omega})}\\
&=\sum_{i} p_{i} \log \Tr\cwich{p_{i}^{-1}\Pi_{i}\overline{\mc{R}}_{{\gamma},{\sigma},\mc{M}}({\omega})} 
=\sum_{i} p_{i} \swich{\log \Tr\cwich{\Pi_{i}\overline{\mc{R}}_{{\gamma},{\sigma},\mc{M}}({\omega})} - \log p_i}\\
 &= D_\Pi(\rho\| \overline{\mc{R}}_{\gamma,\sigma,\mc{M}}(\omega)),  
\end{split}
\end{align} 
where $D_\Pi$ denotes the relative entropy between the distributions obtained on measuring with the eigenbasis of $\rho$.

We now turn to the right hand side of the equality. By concavity of the logarithm, the inequality from~\cite[Corollary 3.3]{sutter16} and the Peierls-Bogoliubov inequality~\cite{araki75,kwon19} that for self-adjoint operators $F,R$ with $\Tr[e^R]=1$ is $\Tr\cwich{e^{F+R}}\geq e^{\Tr\cwich{Fe^R}}$. We apply this to $F=-\log p_i\I+\log{\sigma_{\varepsilon}}+U^\dagger\log{\xi_{\varepsilon}} U+U^\dagger \log{\omega_{\varepsilon}} U$, which is clearly self-adjoint, and $R=\log \Pi_{i,\varepsilon}$, so that $e^R=\Pi_{i,\varepsilon}$ has trace 1.
Therefore
\begin{align}
    \begin{split}
 \sum_i p_i&\log\int_{-\infty}^{\infty}d\beta_0(t)\norm{\prod_{k=1}^5e^{2\alpha H_k}}_2^2  
 \ge  2 \sum_{i} p_{i} \int_{-\infty}^{\infty}d\beta_0(t)\log\norm{\prod_{k=1}^5e^{2\alpha H_k}}_2 \\
\geq&2\sum_{i} p_{i}\log\norm{\exp(\sum_{k=1}^5 H_k)}_2\\
 \geq&\sum_{i} p_{i}\Tr\cwich{e^{\log\Pi_{i,\varepsilon}}\swich{-\log\I p_{i}+\log{\sigma_{\varepsilon}}+U^\dagger\log{\xi_{\varepsilon}} U+U^\dagger \log{\omega_{\varepsilon}} U}} \\
 =&\sum_{i} p_{i}\Tr\cwich{\Pi_{i,\varepsilon}\swich{-\log\I p_{i}+\log{\sigma_{\varepsilon}}+U^\dagger\log{\xi_{\varepsilon}} U+U^\dagger \log\mc{M}(\rho) U-U^\dagger \log\mc{M}(\rho) U+U^\dagger \log{\omega_{\varepsilon}} U}} \\
 =&\sum_{i} p_{i}\Tr\cwich{\Pi_{i,\varepsilon}\swich{-\log\I p_{i}}}\\
 +&\Tr\cwich{\sum_{i} p_{i}\Pi_{i,\varepsilon}\swich{\log{\sigma_{\varepsilon}}+U^\dagger\log{\xi_{\varepsilon}} U+U^\dagger \log\sum_jp_j\mc{M}(\Pi_{j,\varepsilon}) U-U^\dagger \log\sum_jp_j\mc{M}(\Pi_{j,\varepsilon}) U+U^\dagger \log{\omega_{\varepsilon}} U}},
\end{split}
\end{align}
 where the $\sum_jp_j\mc{M}(\Pi_{j,\varepsilon})$ is taken only for $j$ such that $p_j>0$. 
 In order, the first inequality is Jensen's inequality applied to the logarithm, the second inequality uses~\cite[Cor.~3.3]{sutter16}, and the third step applies Peierls-Bogoliubov.

From here we divide this sum into 3 parts to take the limit. First taking the first two terms:

\begin{equation}
    \begin{split}
       \lim_{\varepsilon\to 0} &\sum_{i} p_{i}\swich{-\Tr\cwich{\Pi_{i,\varepsilon}\swich{\log\I p_{i}}} +\Tr\cwich{\Pi_{i,\varepsilon}\log\sigma_\varepsilon}} = \lim_{\varepsilon\to 0}-\sum_i p_i\log p_i +\Tr\cwich{\sum_i p_i\Pi_{i,\varepsilon}\log\sigma_\varepsilon} \\
       &= \lim_{\varepsilon\to 0}-\Tr\cwich{\swich{\sum_{i}p_i\Pi_{i,\varepsilon}}\log \swich{\sum_ip_i\Pi_{i,\varepsilon}}}+\Tr\cwich{\sum_i p_i\Pi_{i,\varepsilon}\log\sigma_\varepsilon}\\
       &=\lim_{\varepsilon\to 0} -D\!\swich{\sum_i p_i\Pi_{i,\varepsilon}\,\middle\|\,\sigma_\varepsilon}\geq -D(\rho\|\sigma),
    \end{split}
\end{equation} where we used that $\Tr\cwich{\Pi_{i,\varepsilon}}=1$, the continuity of the von Neumann entropy and the lower semicontinuity of the relative entropy. Second, we take the third and fourth terms: \begin{equation}
    \begin{split}
    \lim_{\varepsilon\to 0}&\Tr\cwich{\sum_{i} p_{i}\Pi_{i,\varepsilon}\swich{U^\dagger\log{\xi_{\varepsilon}} U+U^\dagger \log\sum_jp_j\mc{M}(\Pi_{j,\varepsilon})U}}\\
    &=\lim_{\varepsilon\to 0}\Tr\cwich{\sum_{i} p_{i}U\Pi_{i,\varepsilon}U^\dagger\swich{\log{\xi_{\varepsilon}} + \log\sum_jp_j\mc{M}(\Pi_{j,\varepsilon})}}\\
    &=\Tr\cwich{\sum_{i} p_{i}U\Pi_{i}U^\dagger\swich{\log{\xi} + \log\mc{M}(\rho) }}
    =\Tr\cwich{\sum_{i} p_{i}U\Pi_{i}U^\dagger\swich{-\log{\gamma} + \log\mc{M}(\rho) }}\\
    &=D(\mc{M}(\rho)\|\gamma),
    \end{split}
\end{equation} where the limit follows from the support condition on $\mc{M}(\rho)$ and $\gamma$, which extends immediately to the $\mc{M}(\Pi_i)$ such that $p_i>0$ and $\xi$, since $\supp \mc{M}(\Pi_i)\subseteq\supp\mc{M}(\rho)$ and $\supp\xi=\supp\gamma$. Finally, we evaluate the fifth and sixth terms: \begin{equation}
    \begin{split}
    \lim_{\varepsilon\to 0}&\Tr\cwich{\sum_{i} p_{i}\Pi_{i,\varepsilon}\swich{-U^\dagger \log\sum_jp_j\mc{M}(\Pi_{j,\varepsilon}) U+U^\dagger \log{\omega_{\varepsilon}} U}}\\
    & =\lim_{\varepsilon\to 0}\Tr\cwich{\sum_{i} p_{i}U\Pi_{i,\varepsilon}U^\dagger\swich{- \log\sum_jp_j\mc{M}(\Pi_{j,\varepsilon}) + \log{\omega_{\varepsilon}} }}\\
    &=\lim_{\varepsilon\to 0} -D\!\swich{\sum_i p_i\mc{M}(\Pi_{i,\varepsilon})\,\middle\|\,\omega_\varepsilon}\geq -D(\mc{M}(\rho)\|\omega),
    \end{split}
\end{equation} 

where we again used the lower semicontinuity of the relative entropy in the limit. 

All together we obtain that the right hand side is larger than

\begin{equation}
-D(\rho\|\sigma)
+
D(\mc{M}(\rho)\|\gamma)
-
D(\mc{M}(\rho)\|\omega).
\end{equation} Joining the left and right hand sides, we obtain the final result (after multiplying both sides by $-1$): \begin{align}
\begin{split}
-D_\Pi\!\left(\rho\middle\|
\overline{\mc{R}}_{\gamma,\sigma,\mc{M}}(\omega)\right)
\geq
-D(\rho\|\sigma)
+
D(\mc{M}(\rho)\|\gamma)
-
D(\mc{M}(\rho)\|\omega).
\end{split}
\end{align}

\end{proof}

\cref{thm:generalEntropyIneq} provides an entropy inequality between arbitrary states. In its most general form the inequality contains too many degrees of freedom, making it difficult to draw direct conclusions. 
 To extract more meaningful statements, we now specialize to particular selections of the auxiliary states $\gamma$ and $\omega$.
 The remainder of the section explores these cases and conclude in the chain rule states in~\cref{coro:chainRule1}.


The following corollary is a weaker version of the result from~\cite{sutter16} (see \cref{eq:strineq} in the introduction) since our result does not involve optimization over measurements. Differently from~\cite{sutter16}, we cannot rely on the variational characterization of the relative entropy, i.e. \begin{equation}
    D(\rho \Vert \sigma)
    = \sup_{H = H^*} \Tr\cwich{ \rho(H) - \log \sigma(e^{H}) },
\end{equation}
which enabled such an optimization in their proof. However, we believe that this restriction arises from technical limitations of the present argument rather than from a fundamental obstruction.
\begin{coro}\label{coro:2channelDPI}
    Let $\mc{M},\mc{N}$ be quantum channels and $\rho$, $\sigma$ states such that $\supp\mc{M}(\rho)\subseteq\supp\mc{N}(\sigma)$. Then \begin{align}
        D(\rho\|\sigma)-D(\mc{M}(\rho)\|\mc{N}(\sigma))\geq D_\Pi(\rho\|\mc{R}_{\mc{N}(\sigma),\sigma,\mc{M}}(\mc{M}(\rho))).
    \end{align} 
\end{coro}

\begin{proof}
    The result is immediate from~\cref{thm:generalEntropyIneq} by letting $\gamma=\mc{N}(\sigma)$ and $\omega=\mc{M}(\rho)$.
\end{proof}

\begin{thm}\label{thm:generalREI}
    Let $\mc{M}$ be a quantum channel, $\rho$, $\sigma$, $\gamma$ states, $\omega$ a positive semidefinite operator, and $D_{\Pi}$ the relative entropy between the distributions obtained when measuring on the basis of $\rho$. Let $\supp \mc{M}(\rho)\subseteq \supp\gamma$. Then
    \begin{align}
    D(\rho\|\sigma)-D(\mc{M}(\rho)\|\gamma)\geq -D(\mc{M}(\rho)\|\omega),
    \end{align}
    if $\Tr\cwich{\Pi_\rho\overline{\mc{R}}_{\gamma,\sigma,\mc{M}}(\omega)} \le 1$.
\end{thm}

\begin{proof}
Consider the setting of~\cref{thm:generalEntropyIneq}. We need to show that \begin{align} 
        T=\Tr\cwich{\Pi_\rho\overline{\mc{R}}_{\gamma,\sigma,\mc{M}}(\omega)}\leq 1 \quad\Rightarrow \quad D_\Pi(\rho\|\overline{\mc{R}}_{\gamma,\sigma,\mc{M}}(\omega))\geq 0.
    \end{align} We can calculate it directly: \begin{align}
        \begin{split}
        D_\Pi(\rho\|\overline{\mc{R}}_{\gamma,\sigma,\mc{M}}(\omega))
        &=D_\Pi\swich{\rho\middle\|T\frac{\overline{\mc{R}}_{\gamma,\sigma,\mc{M}}(\omega)}{T}}\\
        &=D_\Pi\swich{\rho\middle\|\frac{\overline{\mc{R}}_{\gamma,\sigma,\mc{M}}(\omega)}{T}}-\log T\geq 0.
        \end{split}
    \end{align} The first summand is positive because it is the Kullback-Leibler divergence between two probability distributions ($\overline{\mc{R}}_{\gamma,\sigma,\mc{M}}(\omega)/T$ might not be normalised since the trace is in the support of $\rho$, but is in normalised within the support, which is what matters for the entropy of the measured distribution in this basis) and the second one is by the hypothesis. 
\end{proof}

\begin{remark}
    As previously mentioned, the map $\overline{\mc{R}}_{\gamma,\sigma,\mc{M}}$ is completely positive, but not trace preserving (or trace non-increasing for inputs outside the support of $\gamma$) in general. The condition in~\cref{thm:generalREI} reduces to asking whether $\overline{\mc{R}}_{\gamma,\sigma,\mc{M}}$ is trace non-increasing for $\omega$ (for relevant $\omega$, which will be trace 1) within the support of $\rho$. 
    
    For $\gamma=\mc{M}(\sigma)$ we recover the universal recovery channel of \cite{junge18}, which guarantees trace preservation, but we do not recover the full result from \cite{sutter16} due to the specific measurement $\Pi$ in \cref{thm:generalREI}.
\end{remark}

\begin{coro}\label{coro:chainRule1}
    Let $\mc{M}, \mc{N}$ be quantum channels and let $\rho=\sum_j p_j\Pi_j, \sigma$ be states. Then \begin{align}\label{eq:chainRule1}
        D(\rho\|\sigma)-D(\mc{M}(\rho)\|\mc{N}(\sigma))\geq -\mathbb{E}_p D(\mc{M}(\Pi_j)\|\mc{N}(\Pi_j))
    \end{align}
    if $\Tr\cwich{\Pi_\rho\overline{\mc{R}}_{\mc{N}(\sigma),\sigma,\mc{M}}(\mc{N}(\rho))}\leq 1$.
\end{coro}

\begin{proof}
    First we consider the case when $\supp\mc{M}(\rho)\subseteq\supp\mc{N}(\sigma)$. Consider~\cref{thm:generalREI} and let $\omega=\mc{N}(\rho)$ and $\gamma=\mc{N}(\sigma)$. We obtain \begin{align}
        \begin{split}
        D(\rho\|\sigma)&-D(\mc{M}(\rho)\|\mc{N}(\sigma))\geq -D(\mc{M}(\rho)\|\mc{N}(\rho))\\
        &=-D\swich{\mc{M}\swich{\sum_jp_j \Pi_j}\middle\|\mc{N}\swich{\sum_jp_j \Pi_j}}.
        \end{split}
    \end{align}

     We can now use the joint convexity of the quantum relative entropy to find the result: \begin{align}
        \begin{split}
            D(\rho\|\sigma)-D(\mc{M}(\rho)\|\mc{N}(\sigma) )&\geq-D\swich{\mc{M}\swich{\sum_jp_j \Pi_j}\middle\|\mc{N}\swich{\sum_jp_j \Pi_j}}\\
         &\geq -\sum_jp_j D(\mc{M}(\Pi_j)\|\mc{N}(\Pi_j)).
        \end{split}
    \end{align}
          
    Finally let $\supp\mc{M}(\rho)\not\subseteq\supp\mc{N}(\sigma)$. By definition $D(\mc{M}(\rho)\|\mc{N}(\sigma))=+\infty$. We need to show the either right hand side or $D(\rho\|\sigma)$ are infinite. $D(\rho\|\sigma)=+\infty$ if $\supp\rho\not\subseteq\supp\sigma$. Let $\supp\rho\subseteq\supp\sigma$. As $\supp\mc{M}(\rho)\not\subseteq\supp\mc{N}(\sigma)$ by hypothesis, there exists $\Pi_i$ with $p_i>0$ such that $\supp\mc{M}(\Pi_i)\not\subseteq\supp\mc{N}(\sigma)$ due to the linearity of quantum channels. This, together with $\supp\Pi_i\subseteq\supp\rho\subseteq\supp\sigma$, implies that $\supp\mc{M}(\Pi_i)\not\subseteq\supp\mc{N}(\Pi_i)\subseteq\supp\mc{N}(\sigma)$. Because we could choose $p_i$ to be strictly positive, we have $p_iD(\mc{M}(\Pi_i)\|\mc{N}(\Pi_i))=+\infty$, concluding the proof.
\end{proof}

\section{Discussion}\label{sec:discussion}

This work establishes a single-letter chain inequality for quantum relative entropy, \cref{thm:dataProcess1}. In the main text, this result is obtained by lifting the problem to c-q states built from POVM-induced ensemble partitions, applying data processing for the partial trace, and then using measurement monotonicity. Additionally, in \cref{thm:generalEntropyIneq}, we introduce a general entropy inequality from a twisted recovery map in terms of a measured divergence $D_\Pi$. As corollaries, we recover both a conditioned chain rule in which the right hand side maps act on the eigenbasis elements of the initial state $\rho$ (\cref{coro:chainRule1}) and a two-channel strengthened DPI (\cref{coro:2channelDPI}), connecting our results to those ones in~\cite{junge18,sutter16}.

Strengthened data-processing inequalities are expressed in terms of the measured relative entropy with respect to a universal rotated-Petz mixture that exactly recovers the reference state~\cite{junge18,sutter16}. Our twisted recovery  map, by contrast, depends on two reference states $(\gamma,\sigma)$ and is completely positive but not generally trace-non-increasing. 

This construction refines the strengthened data-processing framework by incorporating the recoverability mechanism into a single-letter chain rule, linking entropy loss under a channel to the reconstruction of the input state.

However, unlike~\cite{sutter16}, our proof technique does not optimize over POVMs; thus $D_\Pi$ remains basis-tied to the measurement basis of $\rho$. However we believe that restriction stems from the structure of the proof rather than from any fundamental obstruction of the approach.

Fang, Fawzi, Renner, and Sutter~\cite{fang20} proved that a genuine quantum chain rule holds in infinitely many-copy regime via the regularized channel divergence $D^{\mathrm{reg}}(\mathcal M\|\mathcal N)$. Our results are complementary: they provide informative single-letter bounds that reduce to DPI under commutation assumptions (\cref{example:commutingInput}) and  recover joint convexity (\cref{coro:jointConvexity}), but cannot, in full generality, replace the regularized term.

The POVM-induced partitions $(\rho_j,\sigma_j)$ in \cref{thm:dataProcess1} can be viewed as a link between classical Bayesian updating and quantum conditioning~\cite{fuchs01,leifer13,parzygnat23}. 
From this perspective~\cref{thm:dataProcess1} decomposes the global changes in distinguishability caused by the channels$(\mathcal M,\mathcal N)$ into an average over conditional branches processed by the same channels.

In the commuting case, the bound becomes effectively basis-free on the right-hand side (\cref{example:commutingInput}), reproducing the structural role of point-mass divergences in the classical chain rule.

The chain inequality in \cref{coro:difBasis} depends on the pairing of eigenprojectors; optimizing over permutations would strengthen the bound (Remark~\ref{rem:ordering}). For the conditional chain rule (\cref{coro:chainRule1}), the sufficient trace condition $\Tr[\Pi_\rho\,\mathcal R_{\gamma,\sigma,\mathcal M}(\omega)]\le1$ guarantees a nonnegative $D_\Pi$ term; Appendix~\ref{app:counterexamples} shows that some condition of this form is \emph{necessary} in general, since unconditional formulations fail on explicit families.

Single-letter chain bounds are useful in regimes where reuglarization is infeasible or unnecessary. 
They capture information-theoretic behavior in one-shot or few-shot scenarios—such as hypothesis testing, or channel certification or finite-size security estimates in quantum key distribution. They also offer structural insight in resource-theoretic contexts, where monotones inherit data-processing and convexity properties from relative entropy.

\subsection{Outlook and Open Questions}\label{sec:outlook}

The results in this work suggest a few concrete lines of progress. We group them into three different cathegorises:

\medskip
\noindent\emph{--Choosing the partition and pairing.}
The lower bound in Theorem~\ref{thm:dataProcess1} depends on the POVM that generates the ensemble partitions. A first goal is to characterize an optimal $G^\star(\rho,\sigma;\mathcal M,\mathcal N)$ or to identify families that admit closed forms or efficiently checkable optimality conditions (projective measurements, square-root/Pretty Good Measurement, Naimark-dilated projectors). In the semiclassical variant (Corollary~\ref{coro:difBasis}), the value depends on how eigenprojectors of $\rho$ are paired with those of $\sigma$; the natural problem is to choose a permutation $\pi$ that minimizes
$\sum_j p_j\,D(\mathcal M(\Pi_j)\Vert \mathcal N(\tilde\Pi_{\pi(j)}))$.

\medskip
\noindent\emph{--Sharpening the structure of the bounds.}
Two aspects deserve attention. First, in the twisted-recovery inequality (\cref{thm:generalEntropyIneq},~\cref{coro:2channelDPI}), replacing the basis-tied $D_{\Pi}$ by the measured entropy $\mathbb{D}_M$ would remove what we believe to be a technical artifact of the proof. 
This likely requires a different route (e.g., a variational representation of relative entropy) that preserves the two-reference recovery structure. Second, in \cref{coro:chainRule1} it is natural to ask for necessary and sufficient conditions under which the averaged rotated-Petz map $\overline{\mathcal R}_{\mc{N}(\sigma),\sigma,\mathcal M}$ is trace-non-increasing on a task-relevant cone (for instance, on $\mathrm{range}(\mathcal N)$). It would also be useful to identify regimes of tightness for~\cref{thm:dataProcess1} (commuting inputs, classical-quantum or degradable channels), and to quantify the gap to the many-copy chain rule based on $D^{\mathrm{reg}}(\mathcal M\Vert\mathcal N)$.

\medskip
\noindent\emph{--Extensions.}
A natural next step is to move beyond Umegaki relative entropy and treat Petz- and sandwiched Renyi divergences, as well as hypothesis-testing and smooth min/max divergences, to delineate precisely which divergences admit single-letter chain inequalities (without regularization) and under what support or parameter conditions. In parallel, extending the results to infinite-dimensional settings under energy constraints (e.g., bounded mean energy with respect to a fixed Hamiltonian) would make them applicable to bosonic and continuous-variable systems.

\section*{Acknowledgements}

The authors thank Andreas Winter, Marco Fanizza, John Calsamiglia and Elisabet Roda for helpful discussions.  They also thank the anonimous referees for their helpful comments. 

The authors are supported by the Spanish MICIN 
(project PID2022-141283NB-I00) with the support of FEDER funds, 
by the Spanish MICIN with funding from European Union NextGenerationEU 
(project PRTR-C17.I1) and the Generalitat de Catalunya, and by the Spanish MTDFP 
through the QUANTUM ENIA project: Quantum Spain, funded by the European 
Union NextGenerationEU within the framework of the ``Digital Spain 
2026 Agenda''. 
MHR also thanks the Alexander von Humboldt Foundation and furthermore acknowledges support by the Generalitat de Catalunya through a FI-AGAUR scholarship.
GG also acknowledges financial support by the Stiftung Innovation in der Hochschullehre and the Spanish MICIN (project
PID2022-141283NB-I00) with the support of FEDER funds.

\printbibliography
.
\appendix

\section{Alternative proof of \cref{thm:dataProcess1}}
\label{app:altProof}

In this section we provide an alternative proof of \cref{thm:dataProcess1}. Unlike the main-text proof, which works directly with the c-q lift and the data processing inequality, this proof relies on the classical expressions of the chain rule. It has an extra regularisation step that the proof presented in the main text circumvents. After the regularization we find the elements in the relative entropies separate additvely, which removes the regularisation from the final expression. We note that this proof uses Uhlmann's inequality and an asymptotic equipartition argument in place of the direct partial-trace DPI argument from the main text, highlighting that we do not need DPI for the result. For completion, we also show a simple proof of the classical chain rule using Jensen's Theorem \cite{jensen06,durrett19}.

    
\begin{prop}[Chain rule of the relative entropy]\label{prop:classIneq}
    Let $p, q$ be probability distributions and $M,N$ stochastic maps. Then \begin{align}
        D(p\|q)-D(Mp\|Nq)\geq-\mathbb{E}_p\cwich{D(M\delta_j\|N\delta_j)},
    \end{align} Where $\delta_j$ is the delta probability distribution at point $j$.
\end{prop}

\begin{proof}
    Start by considering two classical probabilities on finite dimension sets and two stochastic maps acting on them. Let $p, q$ be probability distributions and $M,N$ stochastic maps. We consider the output probability distributions \begin{align}
    \tilde{p}_i=(Mq)_i=\sum_jM_{ij}p_j,\quad \tilde{q}_i=(Nq)_i=\sum_jN_{ij}q_j.
\end{align} With this definition we can find the following identity. Consider the quantity \begin{align}
    \exp{-\log{\frac{M_{ij}p_j}{N_{ij}q_j}}+\log{\frac{\tilde{p}_i}{\tilde{q}_i}}}.
\end{align} We calculate the average of this quantity over $M_{ij}p_j$: \begin{align}\begin{split}
    \mathbb{E}_{Mp}{\exp{-\log{\frac{M_{ij}p_j}{N_{ij}q_j}}+\log{\frac{\tilde{p}_i}{\tilde{q}_i}}}}&=\sum_{ij}M_{ij}p_j \exp{-\log{\frac{M_{ij}p_j}{N_{ij}q_j}}+\log{\frac{\tilde{p}_i}{\tilde{q}_i}}} \\
    &=\sum_{ij}M_{ij}p_j\frac{\tilde{p}_i}{\tilde{q}_i}\frac{N_{ij}q_j}{M_{ij}p_j}=\sum_{ij}\frac{\tilde{q}_i}{\tilde{q}_i}\tilde{p}_i=1.
\end{split}\end{align} With this we acquire the following identity for these processes \begin{align}\label{eq:NEpartition}
    \mathbb{E}_{Mp}{\exp{-\log{\frac{M_{ij}p_j}{N_{ij}q_j}}+\log{\frac{\tilde{p}_i}{\tilde{q}_i}}}}=1.
\end{align} From this identity and using Jensen's inequality we can obtain a classical entropy inequality.

    From~\cref{eq:NEpartition} we apply Jensen's inequality:\begin{align}\label{eq:jensenClassical}
        1= \mathbb{E}_{Mp}{\exp{-\log{\frac{M_{ij}p_j}{N_{ij}q_j}}+\log{\frac{\tilde{p}_i}{\tilde{q}_i}}}}\geq \exp{\mathbb{E}_{Mp}{-\log{\frac{M_{ij}p_j}{N_{ij}q_j}}+\log{\frac{\tilde{p}_i}{\tilde{q}_i}}}}.
    \end{align} Note that since $M$ is a stochastic map, for a fixed $j$ $M_{ij}$ is a probability distribution on $i$ and therefore $\sum_i M_{ij}=1$. With this we can calculate the result\begin{align}\begin{split}
        0&\geq \mathbb{E}_{Mp}\swich{-\log{\frac{M_{ij}p_j}{N_{ij}q_j}}+\log{\frac{\tilde{p}_i}{\tilde{q}_i}}}
        =\sum_{ij}M_{ij}p_j\swich{-\log{\frac{M_{ij}p_j}{N_{ij}q_j}}+\log{\frac{\tilde{p}_i}{\tilde{q}_i}}} \\
        &=-\sum_{ij}M_{ij}p_j\cwich{\log{\frac{M_{ij}}{N_{ij}}}+\log{\frac{p_j}{q_j}}}+\sum_i\tilde{p}_i\log{\frac{\tilde{p}_i}{\tilde{q}_i}}\\
        &=-\sum_jp_j\swich{\sum_iM_{ij}\log{\frac{M_{ij}}{N_{ij}}}}-\sum_j\swich{\sum_iM_{ij}}p_j\log{\frac{p_j}{q_j}}+D(Mp\|Nq) \\
        &=-\mathbb{E}_pD(M\delta_j\|N\delta_j)-D(p\|q)+D(Mp\|Nq).
    \end{split}\end{align}
\end{proof}

This proof of our main result,~\cref{thm:dataProcess1}, starts from the states $\omega_{\rho}^{\mc{M}}$ and $\omega_{\sigma}^{\mc{N}}$, introduced in \cref{subsec:bayesian_interpretation} and performs a bipartite measurement to obtain classical probability distributions. These distributions can naturally be interpreted as measurements on $\mc{M}(\rho_{i})$  and $\mc{N}(\sigma_{i})$. The result then follows from using the classical chain rule and removing the measurements with Uhlmann's inequality and an extension to many copies. In the end we obtain the fully quantum, single copy equation.

\begin{proof}[Proof of \cref{thm:dataProcess1}]
    We first construct classical states and maps from our quantum objects.  Given a POVM $F=\{F_i\}$, define 
    \begin{align}
        p_j=P_\rho^G(j)=\Tr\cwich{G_j\rho},\quad q_j=P_\sigma^G(j)=\Tr\cwich{G_j\sigma};
    \end{align} 
    and the stochastic matrices
     \begin{align}
        M_{ij}&=\frac{\Tr\cwich{(G_j\otimes F_i)\omega_\rho^\mc{M}}}{\Tr\cwich{G_j\rho}}=\Tr\cwich{F_i\mc{M}(\rho_j)} =P_{\mc{M}(\rho_j)}^F(i),\\
        N_{ij}&=\frac{\Tr\cwich{(G_j\otimes F_i)\omega_\sigma^\mc{N}}}{\Tr\cwich{G_j\sigma}}=\Tr\cwich{F_i\mc{N}(\sigma_j)}=P_{\mc{N}(\sigma_j)}^F(i).
    \end{align} 
    The corresponding output distributions are
     \begin{align}
        \tilde{p_i}=\sum_{j}M_{ij}p_j=\Tr\cwich{F_i\mc{M}(\rho)}=P_{\mc{M}(\rho)}^F(i),\quad \tilde{q_i}=\sum_{j}N_{ij}q_j=\Tr\cwich{F_i\mc{N}(\sigma)}=P_{\mc{N}(\sigma)}^F(i).
    \end{align}

    We plug these objects into~\cref{prop:classIneq}:
     \begin{align}
        D(P_{\mc{M}(\rho)}^F\|P_{\mc{N}(\sigma)}^F)-D(P_\rho^G\|P_\sigma^G)\leq \mathbb{E}_{P_\rho^G} D(P_{\mc{M}(\rho_j)}^F\|P_{\mc{N}(\sigma_j)}^F),
    \end{align} 
    which in Hayashi's~\cite{hayashi99} notation for the relative entropy for is classical distributions obtained from quantum measurements is 
    \begin{align}\label{eq:has}
        D_F(\mc{M}(\rho)\|\mc{N}(\sigma))-D_G(\rho\|\sigma)\leq \mathbb{E}_{P_\rho^G}D_F(\mc{M}(\rho_j)\|\mc{N}(\sigma_j)).
    \end{align} 

    As a measurement is a particular case of a CPTP map, we can use Uhlmann's ineqaulity \cite{uhlmann77}, the monotonicity of the relative entropy, to remove two of the measurements and obtain the relative entropies between the premeasurement states. Note that the DPI under measurements is a weaker result then the general DPI~\cite{hayashi06,mullerLennert13}, we further comment on this in \cref{rem:DPIproof}
 \begin{align}\label{eq:noNyet}
        D_F(\mc{M}(\rho)\|\mc{N}(\sigma))-D(\rho\|\sigma)\leq \mathbb{E}_{P_\rho^G} D(\mc{M}(\rho_j)\|\mc{N}(\sigma_j)).
    \end{align}

To remove the remaining measurement from the equation, consider $n$ independent copies of the system.
That is, apply (with a small abuse of notation)~\cref{eq:noNyet} to the states $\rho^{\otimes n}$, $\sigma^{\otimes n}$, measured with $G^{\otimes n}$ and processed through the channels $\mc{M}^{\otimes n}, \mc{N}^{\otimes n}$. 
This gives
\begin{align}\label{eq:nineq}
    \begin{split}
     \frac{1}{n}D_F(\mc{M}(\rho)^{\otimes n}\|\mc{N}(\sigma)^{\otimes n})&-\frac{1}{n}D(\rho^{\otimes n}\|\sigma^{\otimes n})\leq \mathbb{E}_{P_{\rho^{\otimes n}}^{G^{\otimes n}}}\frac{1}{n}D\left(\bigotimes_{k=1}^n\mc{M}(\rho_{j_k})\Bigg\|\bigotimes_{k=1}^n\mc{N}(\sigma_{j_k})\right)\\
     &\leq\frac{1}{n}\mathbb{E}_{P_{\rho^{\otimes n}}^{G^{\otimes n}}}\sum_{k=1}^nD(\mc{M}(\rho_{j_k})\|\mc{N}(\sigma_{j_k})) \\
     &= \mathbb{E}_{P_{\rho}^{G}}D(\mc{M}(\rho_{j})\|\mc{N}(\sigma_{j})).
 \end{split}
\end{align}

By additivity $D(\rho^{\otimes n}\|\sigma^{\otimes n})=nD(\rho\|\sigma)$ and by the asymptotic equipartition property~\cite[Theorem 2.3]{hiai91}, we have
\begin{align}
\frac{1}{n}\sup_{F} D_{F}(\mc{M}(\rho)^{\otimes n}\|\mc{N}(\sigma)^{\otimes n})
\;\to\; D(\mc{M}(\rho)\|\mc{N}(\sigma)) \quad \text{as } n\to\infty.
\end{align}
Taking the limit of~\cref{eq:nineq} gives the result~\cref{eq:dataProcessing}.
\end{proof}

\section{Necessity of the condition in~\texorpdfstring{\cref{coro:chainRule1}}{}}\label{app:counterexamples}

In~\cref{coro:chainRule1} we saw a sufficient condition for a chain rule. In this appendix we show an example that demonstrates that there are cases in which the chain rule is false, showing that there exists a necessary condition for the fulfillement of the chain rule. Consider the regularized version of the chain rule:
\begin{align}
    \label{eq:projBound}
    D(\rho\|\sigma ) - D(\mathcal{M}(\rho)\|\mathcal{N}(\sigma))\ge -\lim_{n\to\infty}\frac{1}{n}\mathbb{E}_{\rho_n}\cwich{D(\mathcal{M}^{\otimes n}(\Pi_{k})\|\mathcal{N}^{\otimes n}(\Pi_{k}))},
\end{align} where $\Pi_k$ are the eigenprojectors of$\rho_n=\rho^{\otimes n}$. This version can be obtained by applying~\cref{coro:chainRule1} to $\rho^{\otimes n}$, $\sigma^{\otimes n}$, $\mc{M}^{\otimes n}$ and $\mc{N}^{\otimes n}$, and then taking the limit $n\to\infty$. Therefore a violation of~\cref{eq:projBound} implies a violation of~\cref{eq:chainRule1}.

We provide a simple class of counterexamples in~\cref{ex:projBoundViolation0} and a generalisation of this class in~\cref{ex:projBoundViolationp}, which shows that the states that violate~\cref{eq:projBound} are not some measure 0 set. 

\begin{example}\label{ex:projBoundViolation0}
We show a simple counterexample to~\cref{eq:projBound}. Let $d=2$, $\rho= \ketbra{0}$, $\sigma = (1-\varepsilon)\ketbra{+}+\varepsilon\ketbra{-}$, $\mc{M}(x)=\Tr(x)\ketbra{-}$ and $\mc{N}=\mc{E}_\sigma$, where $\mc{E}_\sigma$ denotes the pinching map on the basis of $\sigma$ \cite{hayashi99}. Then \begin{align}\begin{split}
D(\rho\|\sigma)
&= -\expval{\swich{\log(1-\varepsilon)\ketbra{+}+\log\varepsilon\ketbra{-}}}{0}=-\frac{1}{2}\swich{\log(1-\varepsilon)+\log\varepsilon} \\
D(\mc{M}(\rho)\|\mc{N}(\sigma))
&=D(\ketbra{-}\|\sigma)
=-\expval{\swich{\log(1-\varepsilon)\ketbra{+}+\log\varepsilon\ketbra{-}}}{-}=-\log\varepsilon
\end{split}\end{align} Therefore $D(\rho\|\sigma)-D(\mc{M}(\rho)\|\mc{N}(\sigma))=\frac{1}{2}\swich{\log\varepsilon-\log (1-\varepsilon)},$ which approaches $-\infty$ when $\varepsilon$ goes to 0. We now need to check that the right hand side of~\cref{eq:projBound} is finite to finish the counterexample. 

Consider the projectors of $\rho_n=\rho^{\otimes n}$. Because $\rho$ is pure, $\rho_n$ will have only 2 eigenvalues\footnote{This example also works with $\rho=p\ketbra{0}+(1-p)\ketbra{1}$, $p\neq\frac{1}{2}$, but the resulting 
 projectors and the subsequent calculation are a bit more complicated, see~\cref{ex:projBoundViolationp}.}, 0 and 1, with associated projectors: \begin{align}
\Pi_1=\ketbra{0}^{\otimes n},\quad \Pi_0=\I-\ketbra{0}^{\otimes n}.
\end{align} We need to apply the pinching $\mc{E}_\sigma^{\otimes n}(x)=\sum_{i_1,\dots, i_n\in\set{+,-}}\ketbra{i_1\dots i_n}x\ketbra{i_1\dots i_n}$ to these projectors. Because $\mc{M}^{\otimes n}(\Pi_k)=\ketbra{-}^{\otimes n}$ is pure and $\rho_n$ has a single nonzero eigenvalue, we only care about the value of $\expval{\Pi_1}{-,\dots,-}$ in \begin{align}\mc{N}^{\otimes n}(\Pi_k)=\mc{E}_\sigma^{\otimes n}(\Pi_k)=\sum_{i_1,\dots, i_n\in\set{+,-}}\expval{\Pi_k}{i_1\dots i_n}\ketbra{i_1\dots i_n},\end{align}
 since the term associated to $\Pi_1$ is the only one that matters and the diagonal element of $\mc{N}^{\otimes n}(\Pi_1)$ associated to $\ketbra{-,\dots,-}$ will be the only one that survives in the calculation.

We can calculate this value: \begin{align}\begin{split}
    \expval{\Pi_1}{-\dots-}&=\expval{\ketbra{0}^{\otimes n}}{-\dots -}=\frac{1}{2^n}.
\end{split}
\end{align}

Therefore $D(\mc{M}^{\otimes n}(\Pi_1)\|\mc{N}^{\otimes n}(\Pi_1))=-\log{\frac{1}{2^n}}=n\log{2}=n$. Finally, the right hand side is \begin{align}
    -\lim_{n\to\infty}\frac{1}{n}\mathbb{E}_{\rho_n}[D(\mathcal{M}^{\otimes n}(\Pi_{k})\|\mathcal{N}^{\otimes n}(\Pi_{k}))]=-\lim_{n\to\infty}\frac{1}{n}n=-1>-\infty.
\end{align} A quick calculation shows that $0<\varepsilon<\frac{1}{5}$ violates~\cref{eq:projBound}.
\end{example}


In the following example we generalise by adding an angle to $\sigma$ and a mixture to $\rho$. The resulting set can be seen in~\cref{fig:projBoundViolation}.

\begin{example}\label{ex:projBoundViolationp}
     Let $\rho = (1-p)\ketbra{0}+p\ketbra{1}$, $p\in(0,\fmig)$; $\sigma = (1-\varepsilon)\ketbra{\pe}+\varepsilon\ketbra{\me}$ with $\ket{\pe}=\cos\frac{\theta}{2}\ket{0}+\sin\frac{\theta}{2}\ket{1}$ and $\ket{\me}=\sin\frac{\theta}{2}\ket{0}-\cos\frac{\theta}{2}\ket{1}$; $\mc{M}(x)=\Tr(x)\ketbra{\me}$ and $\mc{N}=\mc{E}_\sigma$. 

We will follow the same steps as in~\cref{ex:projBoundViolation0}. First we calculate the relative entropies: \begin{align}\begin{split}
 D(\rho\|\sigma)=& -S(p)-(1-p)\expval{\swich{\log(1-\varepsilon)\ketbra{\pe}+\log\varepsilon\ketbra{\me}}}{0} \\
 &-p\expval{\swich{\log(1-\varepsilon)\ketbra{\pe}+\log\varepsilon\ketbra{\me}}}{1}\\
 =&-S(p)-(1-p)\cwich{\cos^2\tht\log\swich{1-\varepsilon}+\sin^2\tht\log\varepsilon}\\&-p\cwich{\sin^2\tht\log\swich{1-\varepsilon}+\cos^2\tht\log\varepsilon}\\
 =&-S(p)-\log\swich{1-\varepsilon}\cwich{(1-p)\cos^2\tht+p\sin^2\tht}\\
 &-\log\varepsilon\cwich{(1-p)\sin^2\tht+p\cos^2\tht}\\
 D(\mc{M}(\rho)\|\mc{N}(\sigma))
 =&D(\ketbra{\me}\|\sigma)
 =-\expval{\swich{\log(1-\varepsilon)\ketbra{\pe}+\log\varepsilon\ketbra{\me}}}{\me}\\
 =&-\log\varepsilon
 \end{split}
    \end{align}  Note that the coefficients of $\log(1-\varepsilon)$, $\log\varepsilon$ in $D(\rho\|\sigma)$ add up to 1. Therefore \begin{align}
        D(\rho\|\sigma)-D(\mc{M}(\rho)\|\mc{N}(\sigma))
        = -S(p)+\cwich{\log{\varepsilon}-\log\swich{1-\varepsilon}}\cwich{(1-p)\cos^2\tht+p\sin^2\tht}.
    \end{align} Similarly to in~\cref{ex:projBoundViolation0} $\cwich{\log{\varepsilon}-\log\swich{1-\varepsilon}}$ can be infinitely negtive for small $\varepsilon$ and $\cwich{(1-p)\cos^2\tht+p\sin^2\tht}$ is a strictly positive coefficient, therefore the left hand side of~\cref{eq:projBound} goes to $-\infty$.

$\rho_n$ will now have $n+1$ eigenspaces, each of dimension $\binom{n}{k}$, for $k\in\set{0,\dots,n}$ with eigenvalues $(1-p)^kp^{n-k}$. Let $\mathbb{X}_k=\set{x\in\set{0,1}^n \text{ s.t. }\abs{x}=k},$ where $\abs{x}$ is the number of ones a sequence has. The $k$th projector will be $\Pi_k=\sum_{x\in\mathbb{X}_k}\ketbra{x}$. Similarly to~\cref{ex:projBoundViolation0}, $\mc{N}^{\otimes n}(x)=\mc{E}_\sigma^{\otimes n}(x)=\sum_{i_1,\dots, i_n\in\set{\pe,\me}}\ketbra{i_1\dots i_n}x\ketbra{i_1\dots i_n}$ and we are only concerned with the coefficients associated to $\ketbra{\me}^{\otimes n}$. This coefficients are \begin{align}\begin{split}
   \bra{\me}^{\otimes n}\Pi_k \ket{\me}^{\otimes n} 
    =&\bra{\me}^{\otimes n}\sum_{x\in\mathbb{X}_k}\ketbra{x}\ket{\me}^{\otimes n}
    = \sin^{2k}\tht\cos^{2(n-k)}\tht\abs{\mathbb{X}_k}\\
    =&\binom{n}{k}\sin^{2k}\tht\cos^{2(n-k)}\tht.
\end{split} \end{align} 

The eigenvalue associated to the $k$th eigenspace is $(1-p)^kp^{n-k}$. Due to $\Tr\cwich{\Pi_k}=\binom{n}{k}$, a coefficient $\binom{n}{k}$ appears outside the relative entropy, since $D(k\rho\|k\sigma)=kD(\rho\|\sigma)$. Therefore the expectation value in the right hand side is \begin{align}\label{eq:expvalGeneral}
\begin{split}
-\mathbb{E}_{\rho_n}\cwich{D(\mathcal{M}^{\otimes n}(\Pi_{k})\|\mathcal{N}^{\otimes n}(\Pi_{k}))}
=&\sum_{k=0}^n\binom{n}{k}(1-p)^kp^{n-k}\log\swich{\sin^{2k}\tht\cos^{2(n-k)}\tht}.
\end{split}
\end{align} The limit can be resolved exactly. We can consider first the case $p=0$. While the discussion on the projectors is false for $p=0$ because all there is a single eigenspace and eigenvalue for $k<n$, because this eigenvalue is 0 we can still use~\cref{eq:expvalGeneral}. Only the term $k=n$ survives and it simplifies to $\log\swich{\sin^{2n}\tht}$. The limit is then $\log\swich{\sin^{2}\tht}$ We can find the values of $\varepsilon$ and $\theta$ that violate~\cref{eq:projBound}by setting the left hand side to be smaller than the right hand side. \begin{align}\begin{split}
    \cwich{\log{\varepsilon}-\log\swich{1-\varepsilon}}\cos^2\tht<\log\swich{\sin^{2}\tht} \Leftrightarrow \varepsilon < \frac{(\sin^2\tht)^{\frac{1}{\cos^2\tht}}}{1+(\sin^2\tht)^{\frac{1}{\cos^2\tht}}}
\end{split}
\end{align} These states are plotted in~\cref{fig:projBoundViolation}.

\begin{figure}[ht]
    \centering
      \includegraphics[width=0.49\textwidth]{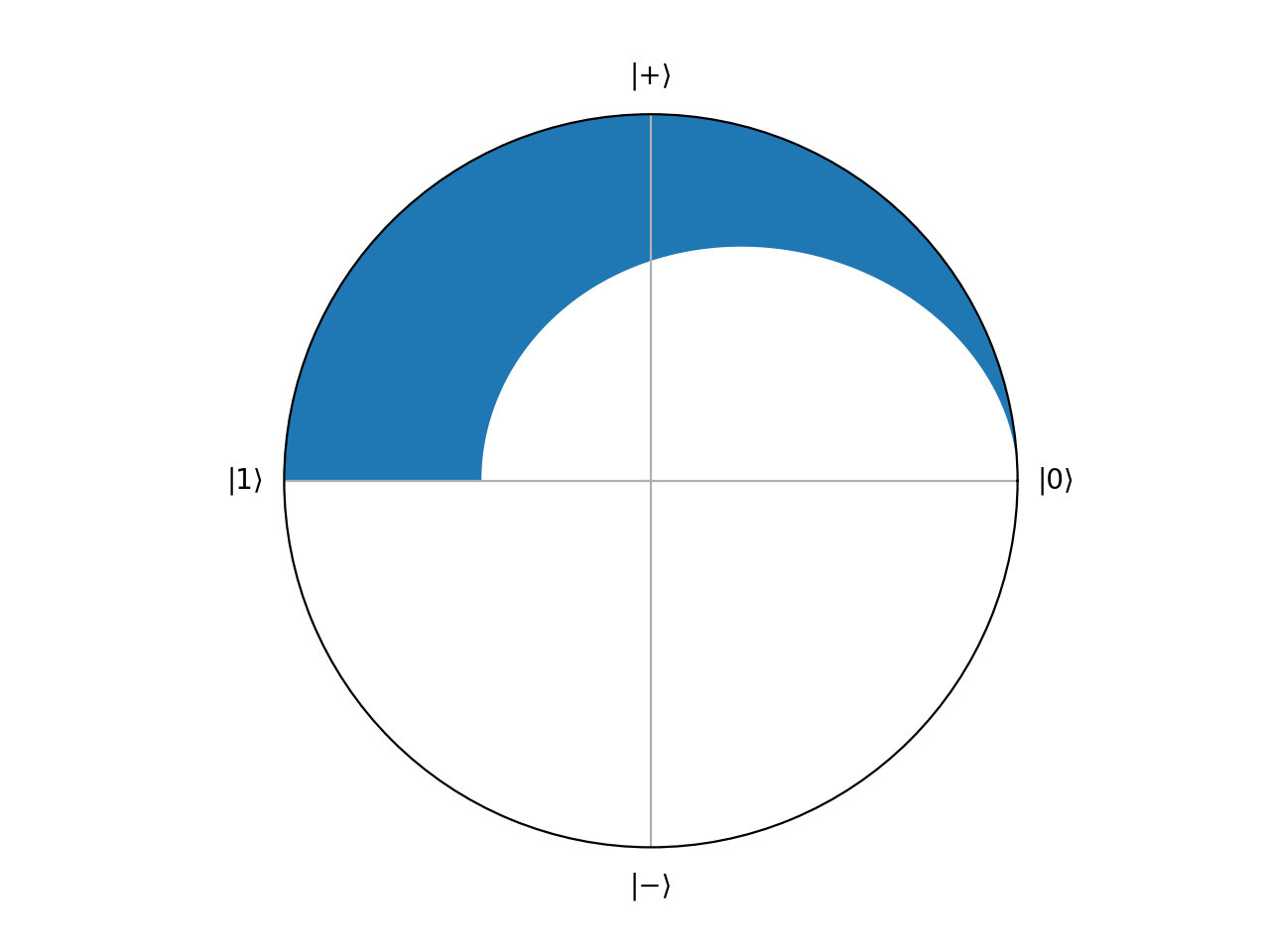}
      \includegraphics[width=0.49\textwidth]{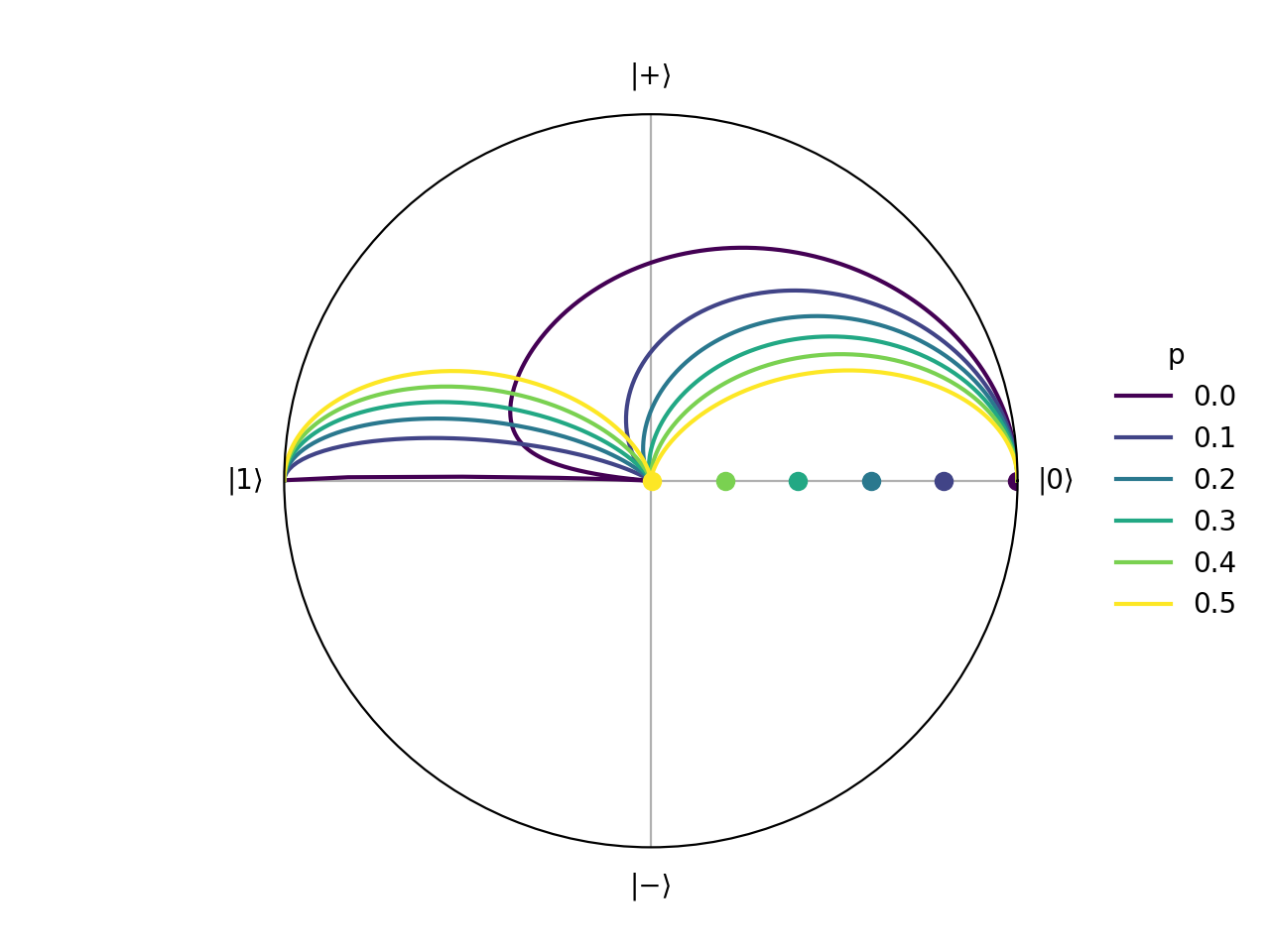}
      \caption{Cross section of the $x-z$ plane of the Bloch sphere. Left shows the states $\sigma$ for which~\cref{eq:projBound} is violated when $\rho=\ketbra{0}$. All the boundaries are not in the set. Right shows the bound for mixed values of $p$. The dots represent the $\rho$ for each value of $p$. Note that in the right picture the cases $p=0$ and $p=\fmig$ are calculated for values very close to $0$ and $\fmig$ but not exaclty $0$ and $\fmig$, since the bound is not continuous at $\fmig$ and at $0$ it has a very rapid change.} 
      \label{fig:projBoundViolation}
    \end{figure}

    If we let $p\in(0,\fmig)$ we can lower bound the right hand side. Note that  \begin{align} \begin{split}
        &\sum_{k=0}^n\binom{n}{k}(1-p)^kp^{n-k}\log\swich{\sin^{2k}\tht\cos^{2(n-k)}\tht} \\
        &=\sum_{k=0}^n\binom{n}{k}(1-p)^kp^{n-k}\swich{k\log\swich{\sin^{2n}\tht}+(n-k)\log{\swich{\cos^{2n}\tht}}} \\
        & =n(1-p)\log\swich{\sin^{2}\tht}+np\log\swich{\cos^{2}\tht} 
        = \log\swich{\sin^{2n(1-p)}\tht\cos^{2np}\tht}
    \end{split}
    \end{align} In the limit this will be $\log\swich{\sin^{2(1-p)}\tht\cos^{2p}\tht}$.~\cref{eq:projBound} will be violated if the left hand side is smaller: \begin{align}\begin{split}
        -S(p)+\cwich{\log{\varepsilon}-\log\swich{1-\varepsilon}}&\cwich{(1-p)\cos^2\tht+p\sin^2\tht}<\log\swich{\sin^{2(1-p)}\tht\cos^{2p}\tht} \\
        &\Updownarrow \\
        \varepsilon &< \frac{\exp{\frac{(1-p)\log\swich{\sin^{2}\tht}+p\log\swich{\cos^{2}\tht}+S(p)}{\cwich{(1-p)\cos^2\tht+p\sin^2\tht}}}}{1+\exp{\frac{(1-p)\log\swich{\sin^{2}\tht}+p\log\swich{\cos^{2}\tht}+S(p)}{\cwich{(1-p)\cos^2\tht+p\sin^2\tht}}}}.
    \end{split}
    \end{align}~\cref{fig:projBoundViolation} shows the bound for different values of $p$. In the limit where $p\to\fmig$ the bound simplifies to \begin{align}
        \varepsilon< \frac{\sin^2\theta}{1+\sin^2\theta}.
    \end{align}
\end{example}

\end{document}